\documentclass[journal]{IEEEtran}

\usepackage[T1]{fontenc}
\usepackage[utf8]{inputenc}
\usepackage[english]{babel}
\usepackage{amsmath, amssymb, amsfonts, bm}
\newcommand{\abs}[1]{\left|#1\right|}
\newcommand{\norm}[1]{\left\lVert #1\right\rVert}

\newcommand{\AM}[1]{\textcolor{black}{#1}}
\usepackage{xcolor}
\newenvironment{anup}
{\begingroup\color{black}}
{\endgroup}

\usepackage{amsthm}

\newtheorem{corollary}{Corollary}
\newtheorem{proposition}{Proposition}
\newtheorem{remark}{Remark}
\renewcommand{\qedsymbol}{}
\usepackage{graphicx}
\usepackage[caption=false,font=footnotesize]{subfig}
\usepackage{tikz}
\usetikzlibrary{matrix, calc, arrows.meta, positioning, angles, quotes, decorations.pathmorphing, shapes.geometric}

\usepackage{pgfplots}
\usepgfplotslibrary{fillbetween} 
\pgfplotsset{compat=1.18}

\usepackage{url}
\usepackage{glossaries}
\usepackage{booktabs}       % For high-quality tables
\usepackage{algorithm}      % For algorithm environment
\usepackage{algpseudocode}  % For Cormen-style pseudocode
\usepackage{siunitx}        % For SI units (e.g., \SI{10}{\hertz})
\usepackage{empheq}         % For highlighting equations (e.g., with boxes)
\usepackage{footnote}       % Better control of footnotes
\usepackage{microtype}      % Improves text spacing and justification
\usepackage{setspace}       % For line spacing control
\usepackage{enumitem}       % For customizing lists
\usepackage{soul}           % For underlining, highlighting, etc. (use sparingly)
\usepackage{comment}        % For comment blocks
\usepackage{tabularx}
\usepackage{makecell}
\usepackage{subfig}
\usepackage{stfloats}
\usepackage{cite}

\newacronym{3gpp}{3GPP}{3rd generation partnership project}
\newacronym{5g}{5G}{5th generation}
\newacronym{6g}{6G}{sixth-generation}
\newacronym{ack}{ACK}{acknowledgement}
\newacronym{afdm}{AFDM}{affine frequency division multiplexing}
\newacronym{aoi}{AoI}{age-of-information}
\newacronym{ap}{AP}{access point}
\newacronym{awgn}{AWGN}{additive white Gaussian noise}
\newacronym{b5g}{B5G}{beyond 5G}
\newacronym{ber}{BER}{bit-error rate}
\newacronym{bp}{BP}{belief propagation}
\newacronym{bs}{BS}{base station}
\newacronym{cscg}{CSCG}{circularly symmetric complex Gaussian}
\newacronym{csi}{CSI}{channel state information}
\newacronym{dd}{DD}{delay--doppler}
\newacronym{dir}{DIR}{data information rate}
\newacronym{dof}{DoF}{degrees-of-freedom}
\newacronym{drl}{DRL}{deep rl}
\newacronym{dtmc}{DTMC}{discrete-time markov chain}
\newacronym{embb}{eMBB}{enhanced mobile broadband}
\newacronym{fdma}{FDMA}{frequency division multiple access}
\newacronym{fec}{FEC}{forward error correction}
\newacronym{fifo}{FIFO}{first-in first out}
\newacronym{fim}{FIM}{fisher information matrix}
\newacronym{ga}{GA}{genetic algorithm}
\newacronym{gap}{GAP}{generalized assignment problem}
\newacronym{geo}{GEO}{geosynchronous earth orbit}
\newacronym{gsl}{GSL}{ground-to-satellite link}
\newacronym{icsi}{ICSI}{imperfect channel state information}
\newacronym{irsa}{IRSA}{irregular repetition slotted aloha}
\newacronym{isa}{ISA}{isolated spectral allocation}
\newacronym{isac}{ISAC}{integrated sensing and communication}
\newacronym{isic}{ISIC}{imperfect successive interference cancellation}
\newacronym{isl}{ISL}{inter-satellite link}
\newacronym{iot}{IoT}{internet of things}
\newacronym{leo}{LEO}{low earth orbit}
\newacronym{lmmse}{LMMSE}{linear minimum mean square error}
\newacronym{los}{LoS}{line-of-sight}
\newacronym{mab}{MAB}{multi-armed bandit}
\newacronym{mdp}{MDP}{markov decision process}
\newacronym{mgap}{MGAP}{multi-level generalized assignment problem}
\newacronym{mimo}{MIMO}{multiple-input multiple-output}
\newacronym{miso}{MISO}{multiple-input single-output}
\newacronym{mmimo}{mMIMO}{massive mimo}
\newacronym{mmse}{MMSE}{minimum mean square error}
\newacronym{mmtc}{mMTC}{massive machine-type communications}
\newacronym{mp}{MP}{message passing}
\newacronym{mu}{MU}{multi-user}
\newacronym{nack}{NACK}{negative ack}
\newacronym{noma}{NOMA}{non-orthogonal multiple access}
\newacronym{nr}{NR}{new radio}
\newacronym{ntn}{NTN}{non-terrestrial network}
\newacronym{ofdm}{OFDM}{orthogonal frequency division multiplexing}
\newacronym{ofdma}{OFDMA}{orthogonal frequency-division multiple access}
\newacronym{oma}{OMA}{orthogonal multiple access}
\newacronym{otfs}{OTFS}{orthogonal time frequency space}
\newacronym{pdf}{PDF}{probability density function}
\newacronym{pwf}{PWF}{phase warping function}
\newacronym{qos}{QoS}{quality of service}
\newacronym{ran}{RAN}{radio access network}
\newacronym{reir}{REIR}{radar estimation information rate}
\newacronym{rl}{RL}{reinforcement learning}
\newacronym{rms}{RMS}{root mean square}
\newacronym{rs}{RS}{rate-splitting}
\newacronym{rsma}{RSMA}{rate-splitting multiple access}
\newacronym{sagin}{SAGIN}{space--air--ground integrated network}
\newacronym{sc}{SC}{superposition coding}
\newacronym{sdma}{SDMA}{space-division multiple access}
\newacronym{sic}{SIC}{successive interference cancellation}
\newacronym{sinr}{SINR}{signal-to-interference-plus-noise ratio}
\newacronym{siso}{SISO}{single-input single-output}
\newacronym{snr}{SNR}{signal-to-noise ratio}
\newacronym{sp}{SP}{signal processing}
\newacronym{tacae}{TCAE}{time-averaged cost of actuation error}
\newacronym{tare}{TRE}{time-averaged reconstruction error}
\newacronym{tgp}{TGP}{tunable Gaussian pulse}
\newacronym{tdma}{TDMA}{time division multiple access}
\newacronym{tf}{TF}{time--frequency}
\newacronym{uav}{UAV}{unmanned aerial vehicle}
\newacronym{uc}{UC}{update-delivery cost}
\newacronym{ue}{UE}{user equipment}
\newacronym{urllc}{URLLC}{ultra-reliable and low-latency communications}
\newacronym{v2x}{V$2$X}{vehicle-to-everything}

\newacronym{af}{AF}{ambiguity function}
\newacronym{ftn}{FTN}{faster-than-Nyquist}
\newacronym{csit}{CSIT}{channel state information at the transmitter}
\newacronym{csir}{CSIR}{channel state information at the receiver}
\newacronym{acf}{ACF}{auto-correlation function}
\newacronym{crlb}{CRLB}{Cram\'er-Rao lower bound}
\newacronym{rrc}{RRC}{root raised cosine}
\newacronym{ul}{UL}{uplink}
\newacronym{semi-isac}{Semi-ISaC}{semi-integrated sensing and communication}
\newacronym{cu}{CU}{communication user}
\newacronym{rt}{RT}{radar target}
\newacronym{rcs}{RCS}{radar cross-section}
\newacronym{dl}{DL}{downlink}
\newacronym{op}{OP}{outage probability}
\newacronym{wssus}{WSSUS}{wide-sense stationary uncorrelated scattering}
\newacronym{sgp}{SGP}{standard Gaussian pulse}
\newacronym{gbps}{Gbps}{gigabits per second}
\newacronym{oobe}{OOBE}{out-of-band emission}
\newacronym{papr}{PAPR}{peak-to-average power ratio}
\newacronym{ipr}{IPR}{interference power ratio}
\newacronym{ici}{ICI}{inter-carrier interference}
\newacronym{isi}{ISI}{inter-symbol interference}
\newacronym{rr}{RR}{radar receiver}
\newacronym{dmrs}{DMRS}{demodulation reference signals}
\newacronym{pusch}{PUSCH}{physical uplink shared channel}
\newacronym{cp}{CP}{cyclic prefix}
\newacronym{ifft}{IFFT}{inverse fast Fourier transform}
\newacronym{fft}{FFT}{fast Fourier transform}
\newacronym{cfr}{CFR}{channel frequency response}
\newacronym{dft}{DFT}{discrete Fourier transform}
\newacronym{dpi}{DPI}{direct-path interference}
\newacronym{tx}{TX}{transmitter}
\newacronym{1tfde}{1T-FDE}{one-tap frequency-domain equalization}
\newacronym{ta}{TA}{timing advance}
\newacronym{wmmse}{WMMSE}{weighted minimum mean square error}
\newacronym{se}{SE}{spectral efficiency}
\newacronym{dc}{DC}{difference-of-convex}
\newacronym{sca}{SCA}{successive convex approximation}
\newacronym{nphard}{NP-hard}{non-deterministic polynomial-time hard}
\newacronym{lmi}{LMI}{linear matrix inequality}
\newacronym{soc}{SOC}{second-order cone}
\newacronym{sqp}{SQP}{sequential quadratic programming}
\newacronym{kpi}{KPI}{key performance indicator}
\newacronym{fr2}{FR2}{frequency range 2}
\newacronym{fdm}{FDM}{frequency division multiplexing}
\newacronym{rmse}{RMSE}{root mean square error}
\newacronym{kkt}{KKT}{Karush-Kuhn-Tucker}
\newacronym{tdl}{TDL}{tapped delay line}
\newacronym{adc}{ADC}{analog-to-digital converter}
\newacronym{isd}{ISD}{inter-site distance}
\newacronym{v2n}{V2N}{vehicle-to-network}
\newacronym{upa}{UPA}{uniform planar array}
\newacronym{uma}{UMa}{urban macro}
\newacronym{nlos}{NLOS}{non-line-of-sight}
\newacronym{cpi}{CPI}{coherent processing interval}
\newacronym{mmwave}{mmWave}{millimeter wave}
\newacronym{idft}{IDFT}{inverse discrete Fourier transform}
\newacronym{td}{TD}{time-domain}
\newacronym{fd}{FD}{frequency-domain}
\newacronym{dp}{DP}{direct-path}
\newacronym{ep}{EP}{echo-path}
\newacronym{mrc}{MRC}{maximum-ratio combining}
\newacronym{mse}{MSE}{mean square error}
\newacronym{bcd}{BCD}{block coordinate descent}
\newacronym{fpp}{FPP}{Feasible Point Pursuit}
\newacronym{qcqp}{QCQPs}{Quadratically constrained quadratic programs}
\newacronym{fp}{FP}{fractional programming}
\newacronym{sdp}{SDP}{semidefinite program}
\newacronym{ao}{AO}{alternating optimization}
\newacronym{psd}{PSD}{positive semi-definite}
\newacronym{ifi}{IFI}{inter-functionality interference}
\newacronym{iid}{i.i.d}{independent and identically distributed}
\newacronym{mi}{MI}{mutual-information}
\newacronym{fdsac}{FDSAC}{frequency domain S$\&$C}
\newacronym{sr}{SR}{sensing-rate}
\newacronym{cr}{CR}{communication-rate}
\newacronym{cc}{C-C}{communication-centric}
\newacronym{stc}{S-C}{sensing-centric}
\newacronym{mac}{MAC}{multiple access channel}
\newacronym{sandc}{S$\&$C}{sensing and communication}
\newacronym{upw}{UPW}{uniform planar wave}
\newacronym{usw}{USW}{uniform spherical wave}
\newacronym{nusw}{NUSW}{non-uniform spherical wave}
\newacronym{xr}{XR}{extended reality}
\newacronym{st}{ST}{sensing target} 

\begin{document}

\title{\AM{Rate-Splitting-Inspired Uplink ISAC: A Rate-Region Analysis}}

\author{{Anup~Mishra},~\IEEEmembership{Member,~IEEE,}
        {Israel~Leyva-Mayorga},~\IEEEmembership{Member,~IEEE,} and~{Petar~Popovski},~\IEEEmembership{Fellow,~IEEE} \vspace{-0.6cm}%
\thanks{A. Mishra, I. Leyva-Mayorga, and P. Popovski are with the Department of Electronic Systems, Aalborg University,  Denmark (e-mail: anmi@es.aau.dk; ilm@es.aau.dk; petarp@es.aau.dk). \textit{Corresponding author: Anup Mishra}.}}

\maketitle

\glsresetall
\begin{anup}
\begin{abstract}
\Gls{isac} enables \gls{sandc} functionalities to share spectrum, hardware, and signal-processing resources, but the resulting inter-functionality interference creates a fundamental receiver-design challenge in uplink operation. To this end, \gls{rs}-inspired \gls{isac} has been proposed as a flexible approach to inter-functionality interference management, whereby communication interference during sensing is
partially decoded and cancelled and partially treated as noise. In this work, we characterize the Pareto-optimal \gls{cr}--\gls{sr} rate region of \gls{rs}-inspired uplink \gls{isac} by jointly optimizing the sensing
illumination and communication-message split. Closed-form \gls{cr} and \gls{sr} expressions are derived while accounting for residual sensing-echo interference caused by target-response estimation uncertainty. We analytically prove that the resulting \gls{rs}-inspired region contains the \gls{noma}-inspired endpoint-order time-sharing region.
We further show that, under fixed sensing illumination, residual
sensing-echo interference can break the containment of the
\gls{oma}-inspired region by both the \gls{rs}- and
\gls{noma}-inspired regions. Joint sensing-illumination optimization
enlarges these non-orthogonal achievable regions and recovers the
maximum \gls{cr} at the zero-\gls{sr} endpoint. High-\gls{snr} and near-field large-array analyses characterize the asymptotic behaviour, and numerical results validate the analysis under both near- and far-field propagation. This is the first work to characterize the Pareto-optimal rate region of \gls{rs}-inspired uplink \gls{isac}.
\end{abstract}
\end{anup}
\begin{IEEEkeywords}
\gls{rs}, \gls{isac}, rate-region analysis.
\end{IEEEkeywords}
\glsresetall
\section{Introduction}
\Gls{isac} has emerged as a key capability for future wireless networks,
enabling \gls{sandc} functionalities to share spectrum, hardware
platforms, and signal processing resources
~\cite{Yuanwei_NOMA_ISaC,chen2025interference,Wei_Survey}. While this
integration improves resource utilization compared with conventional
designs that allocate separate resources to \gls{sandc}, it also creates
inter-functionality interference between the communication signal and
the sensing echo~\cite{Yuanwei_NOMA_ISaC,Zhao2024NearFieldISAC}. This
interference-management problem is particularly relevant in uplink
\gls{isac}, where the \gls{bs} receives the user signal and the target
echo over the same resources and must separate the two functionalities
at the receiver~\cite{Mishra2025RSInspiredISAC,Chiriyath2016}.
\AM{The resulting \gls{sandc} trade-off depends on the
receiver processing order and the spatial coupling between the
\gls{sandc} channels. This coupling further depends on the propagation regime. Under
conventional far-field propagation, the array response is primarily
governed by angular separation, whereas emerging large-aperture,
short-range \gls{isac} settings, including low-altitude/drone tracking
and other mobile-target sensing applications, can operate in the
radiative near field, where spherical-wave propagation introduces
both angular and range-domain structure
~\cite{Zhao2024NearFieldISAC,Cong2023NearFieldISAC,Hua2025NearFieldISAC}. These characteristics
make receiver-side interference management central to the achievable
\gls{sandc} trade-off and motivate the adaptation of
multiple-access principles to uplink \gls{isac}~\cite{Mishra@tutorial,Yuanwei_NOMA_ISaC,Mishra2025RSInspiredISAC}.}
\par \AM{Drawing on} multiple-access techniques for managing inter-user interference in communication-only systems, \gls{isac} studies have adapted orthogonal and non-orthogonal access principles to manage inter-functionality interference between \gls{sandc}~\cite{Chiriyath2016,Yuanwei_NOMA_ISaC}. Allocating separate resources to the two functionalities corresponds to an \gls{oma}-inspired design, with \gls{fdsac} as a representative example~\cite{Yuanwei_Uplink_ISaC1,Zhao2024NearFieldISAC}. \AM{While such orthogonalization avoids inter-functionality interference, it sacrifices the resource-sharing gain of \gls{isac}}, mirroring the spectral-efficiency limitation of orthogonal access in communication-only systems~\cite{Yuanwei_NOMA_ISaC,Mishra2025RSInspiredISAC}. \Gls{noma}-inspired \gls{isac}, in contrast, relies on receiver-side \gls{sic} to manage inter-functionality interference. This naturally yields two endpoint orders: \AM{the \gls{stc} order, where communication is
decoded and removed before sensing, and the \gls{cc} order, where
sensing is performed and the reconstructed echo is suppressed before
communication decoding~\cite{Mishra2025RSInspiredISAC}. Intermediate
\gls{sandc} operating points can then be obtained by time
sharing between these endpoint orders~\cite{Zhao2024NearFieldISAC}.} These orders reveal the \gls{sandc} trade-off, but also inherit a key limitation of the \gls{noma} principle: one functionality may become interference-limited when decoded or estimated in the presence of the other~\cite{Mishra2025RSInspiredISAC,Yuanwei_NOMA_ISaC}.
\par To overcome the limitations of \gls{oma}-inspired and \gls{noma}-inspired designs, \cite{Mishra2025RSInspiredISAC} introduced a \gls{rs}-inspired \gls{isac} framework that adapts the \gls{rs} principle to inter-functionality interference management. In communication-only systems, \gls{rs} splits a message into multiple parts, enabling inter-user interference to be partially decoded and partially treated as noise~\cite{Mishra2022}. In \gls{isac}, this principle is translated to the \gls{sandc} domain: the sensing functionality is recovered under a partial-cancellation structure, where part of the communication signal is decoded and removed before sensing and the remaining part is treated as noise. \AM{Hence, \gls{rs}-inspired \gls{isac} generalizes the two \gls{noma}-inspired endpoint processing orders through the  message-splitting mechanism, while preserving simultaneous use of the
shared \gls{sandc} resources~\cite{Mishra2025RSInspiredISAC}.}
\par \AM{Building on~\cite{Mishra2025RSInspiredISAC}, this paper characterizes the fundamental \gls{sandc} rate-region structure of \gls{rs}-inspired uplink \gls{isac}. We characterize the Pareto-optimal \gls{cr}--\gls{sr} trade-off through joint design of the sensing illumination and communication-message split, and establish its
relationship with the \gls{oma}- and \gls{noma}-inspired achievable regions. The analysis considers both near- and far-field propagation, together with high-\gls{snr} and  near-field large-array characterizations. We next position these developments within the existing literature.}
\vspace{-0.2cm}
\subsection{Related Works}
\par Early coexistence work characterized achievable
\gls{sandc} trade-offs through orthogonal sub-band and \gls{sic}-based inner bounds~\cite{Chiriyath2016}. Building on this
coexistence perspective, \cite{Zhang2023SemiISAC}
introduced a semi-\gls{isac} framework that partitions the bandwidth into communication-only, sensing-only, and mixed \gls{isac} blocks, and evaluated \gls{oma}- and \gls{noma}-inspired implementations
in terms of outage probability, ergodic \gls{cr}, and ergodic \gls{reir}. The role of inter-functionality \gls{sic} ordering was further examined in~\cite{OuyangImpactrevealSIC} through \gls{stc}- and \gls{cc}-based uplink \gls{isac} designs. More recently,
\cite{Mishra2025RSInspiredISAC} introduced an \gls{rs}-inspired uplink \gls{isac} framework that splits the communication message around the sensing operation, thereby providing more flexible inter-functionality
interference management than \gls{oma}- and \gls{noma}-inspired baselines.\footnote{\AM{Related work has also considered downlink \gls{noma}-inspired \gls{isac}, including designs that embed information into dedicated sensing waveforms and jointly optimize \gls{sandc}
performance~\cite{Wang2022NOMAInspired,Yuanwei_NOMA_ISaC}. Mixed near- and far-field propagation in \gls{isac} systems has also
been considered in~\cite{yoo2025exploring} through adaptive bandwidth
splitting. Broader surveys of multiple-access-enabled \gls{isac} and interference
management are provided in~\cite{chen2025interference,Yuanwei_NOMA_ISaC,Longfei2022a}.}}
\par \AM{In terms of the underlying rate-region setting, \cite{Chiriyath2016} derived a \gls{sandc} inner bound under a fixed \gls{stc} processing order. Within an uplink \gls{isac}
setting, \cite{Yuanwei_Uplink_ISaC1} likewise adopted a fixed \gls{stc} order, while optimizing the sensing waveform within that order. Reference \cite{OuyangImpactrevealSIC} subsequently considered both \gls{stc} and \gls{cc} orders and obtained intermediate \gls{sandc} operating points through time sharing between the two endpoint orders, with ideal sensing-echo removal assumed for the \gls{cc} receiver. Extending this framework to near-field propagation,~\cite{Zhao2024NearFieldISAC} incorporated the range- and angle-dependent channel structure induced by spherical-wave propagation, while retaining the endpoint-order time-sharing construction and ideal sensing-echo removal for the \gls{cc} receiver under a fixed sensing-matched beamformer. Such perfect sensing-echo removal is generally unattainable with finite
sensing resources, since target-response estimation errors leave residual echo-reconstruction mismatch that propagates into the post-sensing communication stage~\cite{Kay1993,Mishra2025RSInspiredISAC}. In contrast, the \gls{rs}-inspired framework in~\cite{Mishra2025RSInspiredISAC} introduced communication-message splitting around the sensing operation, thereby enabling flexible inter-functionality interference management, and accounted for the residual sensing-echo interference. The resulting \gls{sandc} performance trade-off was studied under a fixed sensing configuration, with the communication-message split optimized for communication performance.}
\par \AM{However, the Pareto-optimal
\gls{cr}--\gls{sr} region of \gls{rs}-inspired uplink \gls{isac},
obtained through the joint optimization of the sensing illumination and
communication-message split, remains uncharacterized. Consequently, the
relationship of this region with the \gls{oma}- and
\gls{noma}-inspired achievable regions is also unknown. This relationship
becomes particularly important once residual sensing-echo interference
due to finite target-response estimation accuracy is accounted for,
since the resulting coupling between sensing estimation and subsequent
communication decoding can fundamentally alter the relationship between
orthogonal and non-orthogonal rate regions relative to the
ideal-cancellation case. Beyond
the finite-\gls{snr} rate region, the corresponding high-\gls{snr}
behaviour and, under near-field propagation, the large-array limits of
the \gls{rs}-inspired architecture also remain unexplored.}
\vspace{-0.2cm}
\subsection{Contributions}
\par \AM{Motivated by these gaps, this paper develops a rate-region analysis of residual-aware \gls{rs}-inspired uplink \gls{isac} by jointly optimizing the sensing illumination and communication-message split. The analysis encompasses both far- and near-field propagation,
together with high-\gls{snr} and near-field large-array characterizations. Our main contributions are as follows.}
\begin{itemize}
    \item \AM{Building on the \gls{rs}-inspired receiver architecture, we derive closed-form \gls{cr} and \gls{sr} expressions for \gls{rs}-inspired uplink \gls{isac} as functions of the communication-message split and sensing illumination. The resulting
    expressions recover the \gls{stc} and \gls{cc} endpoint orders as   limiting cases and explicitly propagate target-response estimation  uncertainty into the post-sensing communication stage, revealing how the residual interference depends on sensing accuracy, sensing illumination, and the overlap between the \gls{sandc} channels.}

    \item \AM{This is the first work to characterize the Pareto-optimal \gls{cr}--\gls{sr} rate region of \gls{rs}-inspired uplink \gls{isac} by jointly optimizing the sensing illumination and communication-message split. Moreover, we analytically prove that the resulting \gls{rs}-inspired region contains the \gls{noma}-inspired endpoint-order time-sharing region. This differs fundamentally from the classical fixed-power Gaussian uplink \gls{mac}, where varying the \gls{rs} split only traverses the same dominant face achievable by time sharing between \gls{sic} decoding orders. We further show that, once residual sensing-echo interference is accounted for, a \gls{noma}-inspired region with fixed sensing illumination need not contain the \gls{oma}-inspired region, unlike under ideal sensing-echo cancellation. Joint optimization of the sensing illumination enlarges the achievable region and recovers the maximum \gls{cr} for non-orthogonal schemes at the zero-\gls{sr} endpoint.}

    \item \AM{We derive high-\gls{snr} and near-field large-array asymptotic characterizations for \gls{rs}-inspired uplink \gls{isac}. The high-\gls{snr} analysis shows that, for non-aligned \gls{sandc} channels, residual sensing interference affects the rate offsets but not the leading \gls{cr} and \gls{sr} slopes, whereas under full channel alignment it becomes slope-limiting for the affected \gls{sandc} stages. The large-array analysis establishes that, under aperture-aware near-field propagation, the \gls{rs}-inspired scheme, like its \gls{noma}-inspired counterparts, approaches finite large-array rate limits.}

    \item \AM{We validate the rate-region and asymptotic analyses numerically across near- and far-field propagation. The results show that \gls{rs}-inspired message splitting achieves \gls{sandc} operating points beyond \gls{noma}-inspired endpoint-order time sharing, quantify the communication loss caused by residual sensing-echo interference relative to ideal cancellation, particularly under stronger channel alignment, and confirm the finite large-array \gls{cr} and \gls{sr} limits predicted by the near-field analysis.}
\end{itemize}
\AM{These results establish the \gls{rs}-inspired uplink \gls{isac} rate
region as a generalization of the \gls{noma}-inspired rate region and
reveal how residual sensing interference alters its relationship with
the corresponding \gls{oma}-inspired benchmark. As in the Gaussian uplink \gls{mac}, \gls{rs} realizes the \gls{noma}-inspired time-sharing region within a single frame, avoiding inter-frame scheduling and decoding-order switching, while sensing-dependent coupling can further enlarge the achievable region.}
\par \textit{Notation:} Scalars, vectors, and matrices are denoted by lower-case, bold lower-case, and bold upper-case letters, respectively. The transpose, conjugate, and Hermitian transpose are denoted by \((\cdot)^T\), \((\cdot)^*\), and \((\cdot)^H\), respectively. The Euclidean norm and absolute value are denoted by \(\norm{\cdot}\) and \(\abs{\cdot}\). The \(N\times N\) identity matrix is denoted by \(\mathbf I_N\). The operators \(\mathbb{E}[\cdot]\), \(\det(\cdot)\), and \(\mathrm{vec}(\cdot)\) denote expectation, determinant, and vectorization, respectively. The Kronecker product is denoted by \(\otimes\), and \(\mathcal{CN}(\boldsymbol{\mu},\mathbf{C})\) denotes a \gls{cscg} with mean \(\boldsymbol{\mu}\) and covariance \(\mathbf{C}\).
\par \textit{Organization:} The remainder of this paper is organized as
follows. Section~\ref{sec:sysmod} presents the uplink \gls{isac} system
and signal models under near- and far-field propagation.
Section~\ref{sec:rsisac} derives the \gls{cr} and \gls{sr} expressions
for the \gls{rs}-inspired receiver and develops the high-\gls{snr} and
large-array analyses. Section~\ref{sec:rate_region_characterization}
characterizes the achievable \gls{cr}--\gls{sr} rate region and
establishes its relationship with the \gls{oma}- and
\gls{noma}-inspired benchmarks. Section~\ref{sec:num_results} provides
numerical results, and Section~\ref{sec:conclusion} concludes the paper.
Proofs of the main analytical results are provided in the appendices.

\section{System Model}
\label{sec:sysmod}
\begin{anup}
We consider an uplink \gls{isac} system comprising a dual-functional
\gls{bs} equipped with an $N$-element \gls{upa}, where $N=N_yN_z$.
The \gls{upa} is centered at the origin and deployed on the $y$--$z$
plane. Let $d$ denote the inter-element spacing and let each antenna
element occupy an area $A$. The corresponding array occupation ratio is defined as $\zeta \triangleq A/d^2 \in (0,1]$~\cite{Zhao2024NearFieldISAC}. A single-antenna \gls{cu} ($i=c$) and a single \gls{st} ($i=s$)
are considered.\footnote{\AM{The single-\gls{cu}, single-\gls{st} setup follows prior \gls{isac}
rate-region studies~\cite{Chiriyath2016,Zhang2023SemiISAC,Zhao2024NearFieldISAC}
and enables a tractable characterization of the fundamental
\gls{cr}--\gls{sr} trade-off. Extensions to multiple \glspl{cu} or \glspl{st} would introduce additional message-splitting, \gls{sic}, estimation, and residual interference interactions, while extended targets would require a
higher-dimensional target-response model in place of the scalar
coefficient $\beta$. These extensions lead to higher-dimensional
rate-region problems and are left for future work.}} Further, we consider both near- and far-field propagation
~\cite{Qu2023NearFieldISAC,Luo2023BeamSquint,Elbir2023NearField}.
For the near-field regime, we adopt the aperture-aware spherical-wave
array-channel formulation in~\cite{Zhao2024NearFieldISAC}, which yields
physically consistent large-array behaviour
~\cite{Lu2022XLArray,Dardari2020LIS,Bjornson2020PowerScaling}, while the
far-field regime is modelled using the conventional \gls{upw} response~\cite{Zhao2024NearFieldISAC,Wang2023NearFieldISAC}.
\par For $i\in\{c,s\}$, the location of node $i$ is parameterized by its
distance $r_i$ from the array center, elevation angle $\theta_i$, and
azimuth angle $\phi_i$. Define
\begin{equation}
\Psi_i \triangleq \sin\theta_i\cos\phi_i,\quad
\Phi_i \triangleq \sin\theta_i\sin\phi_i,\quad
\Omega_i \triangleq \cos\theta_i,
\end{equation}
such that the Cartesian position vector of node $i$ is
\begin{equation}
    \mathbf r_i
    =
    [r_i\Psi_i,\;r_i\Phi_i,\;r_i\Omega_i]^T.
\end{equation}
\begin{figure}[t]
    \centering
    \includegraphics[width=0.70\columnwidth]{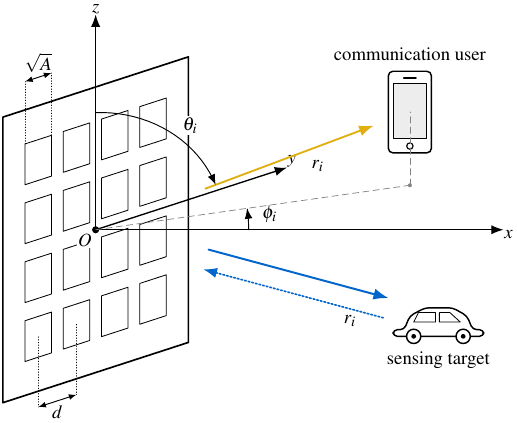}
    \caption{Uplink \gls{isac} system model and \gls{upa} geometry with
    the \gls{cu}, \gls{st}, and corresponding
    \gls{sandc} links.\vspace{-0.5cm}}
    \label{fig:system_model}
\end{figure}
The center of the $(n_y,n_z)$-th antenna element is given by
\begin{equation}
    \mathbf p_{n_y,n_z}
    =
    [0,\;n_y d,\;n_z d]^T,
\end{equation}
where
$n_y\in\{0,\pm1,\ldots,\pm(N_y-1)/2\}$ and
$n_z\in\{0,\pm1,\ldots,\pm(N_z-1)/2\}$.
Fig.~\ref{fig:system_model} illustrates the system model.
\par Under near-field propagation, the distance between node $i$ and the
$(n_y,n_z)$-th array element is
\begin{align}
r_{n_y,n_z,i}
&=
r_i\sqrt{
(n_y\epsilon_i-\Phi_i)^2
+
(n_z\epsilon_i-\Omega_i)^2
+
\Psi_i^2
},
\label{eq:nf_distance}
\end{align}
where $\epsilon_i\triangleq d/r_i$. Under far-field propagation, the
spherical wavefront is approximated by a \gls{upw}. Following~\cite{Zhao2024NearFieldISAC,Wang2023NearFieldISAC}, the corresponding
near- and far-field element-wise channel responses are, respectively,
\begin{align}
h_{n_y,n_z,i}^{\rm NF}
&=
\sqrt{
\frac{
A\!\left(
r_i^3\Psi_i^3
+
r_i\Psi_i(r_i\Omega_i-n_zd)^2
\right)
}{
4\pi r_{n_y,n_z,i}^5
}
}
e^{-j\frac{2\pi}{\lambda}r_{n_y,n_z,i}},
\nonumber
\\
h_{n_y,n_z,i}^{\rm FF}
&=
\sqrt{\frac{A}{4\pi r_i^2}}\,
e^{-j\frac{2\pi}{\lambda}
\left(r_i-n_yd\Phi_i-n_zd\Omega_i\right)}.
\label{eq:nf_ff_element_channel}
\end{align}
The adopted near-field response follows the simplified aligned-polarization
configuration in~\cite{Zhao2024NearFieldISAC} and captures spherical-wave
phase and element-dependent amplitude due to propagation, effective-aperture,
and polarization losses. In contrast,
the \gls{usw} model assumes uniform amplitude~\cite{Wang2023NearFieldISAC},
while the far-field \gls{upw} model further adopts linear phase
progression across the array.
\par By stacking the corresponding element-wise coefficients in
\eqref{eq:nf_ff_element_channel}, we obtain the near- and far-field
channel vectors $\mathbf h_i^{\rm NF}$ and $\mathbf h_i^{\rm FF}$,
respectively, for $i\in\{c,s\}$. Hereafter, $\mathbf h_i$ denotes the
channel vector under the propagation regime being considered, so that
the subsequent rate-region analysis applies to both regimes. The spatial
overlap between the \gls{sandc} channels is characterized by
\begin{equation}
\label{eq:rho}
\rho \triangleq
\frac{\abs{\mathbf h_c^H\mathbf h_s}^2}
{\norm{\mathbf h_c}^2\norm{\mathbf h_s}^2}
\in[0,1],
\end{equation}
where $\rho=0$ and $\rho=1$ correspond to orthogonal and fully aligned
channel vectors, respectively. Thus, the propagation regime enters
the subsequent analysis through $\norm{\mathbf h_c}^2$,
$\norm{\mathbf h_s}^2$, and overlap factor $\rho$.

\subsection{Signal Model}
We consider an uplink frame of length $L$ symbols, with
$\mathbf{h}_s,\mathbf{h}_c\in\mathbb{C}^{N\times 1}$ denoting the
\gls{sandc} channel vectors, respectively, under the
propagation regime being considered. Throughout, the receiver noise is modeled as standard \gls{awgn} with independent $\mathcal{CN}(0,1)$ entries.

\subsubsection{Monostatic Sensing Signal}
The \gls{bs} transmits a unit-modulus sensing pulse sequence
$\mathbf{s}_s\in\mathbb{C}^{L\times 1}$ satisfying $\abs{s_{s,\ell}}^2=1,\ \forall \ell\in\{1,\ldots,L\}$, through a normalized sensing beamformer
$\mathbf{w}\in\mathbb{C}^{N\times 1}$ with $\norm{\mathbf{w}}^2=1$.
Accordingly, the transmitted sensing signal matrix is
\begin{equation}
\mathbf{X}_s
=
\sqrt{p_s}\,\mathbf{w}\mathbf{s}_s^H
\in\mathbb{C}^{N\times L},
\end{equation}
where $p_s\in\mathbb{R}_+$ denotes the sensing transmit power. Under the
monostatic single-target model, the target response matrix is written
as
\begin{equation}
\mathbf{G}
=
\beta\mathbf{h}_s\mathbf{h}_s^T
\in\mathbb{C}^{N\times N},
\end{equation}
where $\beta\sim\mathcal{CN}(0,\alpha_s)$ is the unknown complex target
reflection coefficient and $\alpha_s\in\mathbb{R}_+$ denotes its average
power. Following~\cite{Zhao2024NearFieldISAC}, we consider a tracked-target
setting in which the target geometry determining $\mathbf h_s$ is
available at the \gls{bs} over the considered frame, while the complex
target-response coefficient $\beta$ remains unknown and is estimated
from the received sensing echo. Accordingly, the sensing task considered
here corresponds to tracking/monitoring after target acquisition, rather
than initial target detection or localization. The sensing echo
component received at the \gls{bs} over one frame is therefore
\begin{equation}
\sqrt{p_s}\,\mathbf{G}\mathbf{w}\mathbf{s}_s^H
=
\sqrt{p_s}\,\beta\,\mathbf{h}_s\mathbf{h}_s^T
\mathbf{w}\mathbf{s}_s^H
\in\mathbb{C}^{N\times L}.
\end{equation}
\subsubsection{\gls{rs}-Inspired Uplink Communication Signal}
The single-antenna \gls{cu} employs \gls{rs} and decomposes its message into two independent Gaussian streams $\mathbf{s}_{c,1},\mathbf{s}_{c,2}\in\mathbb{C}^{L\times 1}$ satisfying $\mathbb{E}\!\left[\mathbf{s}_{c,1}\mathbf{s}_{c,1}^H\right]
=
\mathbb{E}\!\left[\mathbf{s}_{c,2}\mathbf{s}_{c,2}^H\right]
=
\mathbf{I}_L.$ The corresponding transmit powers are given by $p_{c,1}=(1-\alpha)p_c$ and $p_{c,2}=\alpha p_c,$ where $p_c\in\mathbb{R}_+$ denotes the total uplink communication power and $\alpha\in[0,1]$ is the \gls{rs} split factor~\cite{Mishra2025RSInspiredISAC}.\footnote{\AM{As a scalar control
parameter, $\alpha$ can be conveyed to the \gls{cu} through conventional
control signalling\cite{Mishra2025RSInspiredISAC,Aditya_RSMA_experimental}.}} The transmitted communication signal over one frame is therefore
\begin{equation}
\mathbf{X}_c
=
\sqrt{p_{c,1}}\,\mathbf{s}_{c,1}^H
+
\sqrt{p_{c,2}}\,\mathbf{s}_{c,2}^H
\in\mathbb{C}^{1\times L},
\end{equation}
and the corresponding received signal at the \gls{bs} is given by
\begin{equation}
\mathbf{h}_c\mathbf{X}_c
=
\sqrt{p_{c,1}}\,\mathbf{h}_c\mathbf{s}_{c,1}^H
+
\sqrt{p_{c,2}}\,\mathbf{h}_c\mathbf{s}_{c,2}^H
\in\mathbb{C}^{N\times L}.
\end{equation}
\end{anup}
\vspace{-0.5cm}
\section{\Gls{rs}-Inspired Uplink \gls{isac}}\label{sec:rsisac}
We now derive the \gls{cr} and \gls{sr} of the
\gls{rs}-inspired uplink \gls{isac} receiver.\footnote{\AM{The communication channel $\mathbf h_c$ is assumed to be
available at the \gls{bs} over the considered frame. As
in~\cite{Zhao2024NearFieldISAC}, the target geometry determining
$\mathbf h_s$ is assumed to be tracked and available at the \gls{bs},
while $\beta$ remains unknown and is estimated from the sensing
observation. Imperfect knowledge of $\mathbf h_c$ or uncertainty in
additional target parameters would not alter the prescribed receiver
processing sequence, but would require explicitly modelling the resulting
channel uncertainty and its impact on decoding performance and the
achievable rate region, which is left for future work.}}
\AM{Following~\cite{Mishra2025RSInspiredISAC}, the receiver adopts the decoding order
$s_{c,1}\!\rightarrow\!\beta\!\rightarrow\!s_{c,2}$, thereby allowing
the message split to control how much communication interference is
removed before sensing estimation and how much communication is decoded
after sensing-echo cancellation.} The split factor
$\alpha\in[0,1]$ interpolates between the two endpoint-\gls{sic}
strategies: $\alpha=0$ recovers the \gls{stc} order, while $\alpha=1$
recovers the \gls{cc} order~\cite{OuyangImpactrevealSIC}. The
corresponding \gls{noma}- and \gls{rs}-inspired receiver processing
structures are illustrated in Fig.~\ref{fig:receiver_architectures}. \footnote{\AM{Within a frame, the \gls{rs}-inspired receiver adds one
communication decoding/cancellation stage relative to either unsplit
endpoint receiver, while retaining the same sensing stage. The underlying
message-splitting and successive-decoding operations have practical
precedent in \gls{rsma} systems~\cite{Aditya_RSMA_experimental}.
The \gls{noma}-inspired time-sharing benchmark instead requires switching
between endpoint processing orders across frames.}} Thus, the \gls{rs}-inspired model generalizes the \gls{noma}-inspired
endpoint benchmark in~\cite{Zhao2024NearFieldISAC}. The following
subsections derive the corresponding \gls{cr}, \gls{sr}, asymptotic
behaviour, and rate-region characterization.
\begin{figure}[t]
    \centering
    \includegraphics[width=\columnwidth]{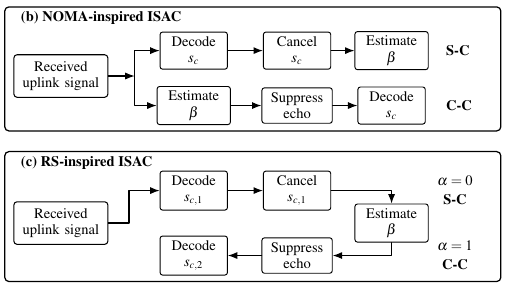}
    \caption{\AM{Receiver processing in \gls{noma}- and \gls{rs}-inspired
    uplink \gls{isac}.}}
    \label{fig:receiver_architectures}\vspace{-0.5cm}
\end{figure}
\vspace{-0.3cm}
\subsection{First-Stream \gls{cr}}
We first consider the decoding of stream $\mathbf{s}_{c,1}$. Under the symbol-level communication decoding considered here,
$\mathbf{h}_c$ and the marginal interference covariance remain
invariant across the $L$ symbols of a frame, so the corresponding rate
can be evaluated on a per-symbol basis~\cite{Zhao2024NearFieldISAC}. At this
stage, the second communication stream and the sensing echo are treated
as interference. Under Gaussian signaling and conditioned on the known
channel realization, the resulting interference-plus-noise covariance
matrix, $\mathbf{R}_1(\mathbf{w},\alpha)\in\mathbb{C}^{N\times N}$, is
given by
\begin{equation}
\label{eq:R1_cov}
\mathbf{R}_1(\mathbf{w},\alpha)
\triangleq
\mathbf{I}_N
+
p_{c,2}\mathbf{h}_c\mathbf{h}_c^H
+
p_s\alpha_s\abs{\mathbf{h}_s^T\mathbf{w}}^2\,\mathbf{h}_s\mathbf{h}_s^H.
\end{equation}
Accordingly, an achievable rate for the first communication stream can be written as
\begin{equation}
\label{eq:R1}
R_{c,1}(\mathbf{w},\alpha)
=
\log_2\det\!\left(
\mathbf{I}_N
+
p_{c,1}\mathbf{h}_c\mathbf{h}_c^H\mathbf{R}_1^{-1}(\mathbf{w},\alpha)
\right).
\end{equation}
Since the desired signal covariance $p_{c,1}\mathbf{h}_c\mathbf{h}_c^H$ is rank one,  \eqref{eq:R1} can be rewritten using the matrix determinant lemma, as
\begin{equation}
\label{eq:R1_scalar}
R_{c,1}(\mathbf{w},\alpha)
=
\log_2\!\left(
1+p_{c,1}\mathbf{h}_c^H\mathbf{R}_1^{-1}(\mathbf{w},\alpha)\mathbf{h}_c
\right).
\end{equation}
\par The following proposition provides an explicit characterization of the first-stream rate.
\begin{proposition}
\label{prop:R1_closed}
Under the Gaussian signaling model above, the achievable rate of the first communication stream admits the closed-form expression in \eqref{eq:R1_closed_form}.
\end{proposition}
\begin{proof}
The derivation follows from a rank-two Woodbury expansion of $\mathbf{R}_1^{-1}(\mathbf{w},\alpha)$ and is provided in Appendix~\ref{app:first_stream_rate}. 
\end{proof}
\begin{figure*}[!htbp]
\vspace{-0.6em}
\begin{equation}
\label{eq:R1_closed_form}
R_{c,1}(\mathbf{w},\alpha)
=
\log_2\!\left(
1+
p_{c,1}
\frac{
\norm{\mathbf{h}_c}^2
\left[
1+
p_s\alpha_s\abs{\mathbf{h}_s^T\mathbf{w}}^2
\norm{\mathbf{h}_s}^2(1-\rho)
\right]
}{
1+
p_{c,2}\norm{\mathbf{h}_c}^2
+
p_s\alpha_s\abs{\mathbf{h}_s^T\mathbf{w}}^2\norm{\mathbf{h}_s}^2
+
p_{c,2}p_s\alpha_s\abs{\mathbf{h}_s^T\mathbf{w}}^2
\norm{\mathbf{h}_c}^2\norm{\mathbf{h}_s}^2(1-\rho)
}
\right).
\end{equation}
\hrulefill
\vspace{-1.2em}
\end{figure*}

\par In particular, when $\alpha=0$, the \gls{rs}-inspired formulation reduces to the \gls{stc} uplink \gls{sic} strategy. Proposition~\ref{prop:R1_closed} admits several useful interpretations. When \(\rho=0\), the \gls{sandc} channels are orthogonal, and \eqref{eq:R1_closed_form} reduces to
\begin{equation}
\label{eq:R1_rho0}
R_{c,1}(\mathbf{w},\alpha)
=
\log_2\!\left(
\frac{
1+p_c\norm{\mathbf{h}_c}^2
}{
1+\alpha p_c\norm{\mathbf{h}_c}^2
}
\right),
\end{equation}
which is independent of the sensing beamformer. At the opposite extreme,
when $\rho=1$, the \gls{sandc} channels are fully aligned, resulting in
the strongest coupling of the sensing echo into the communication
direction. More generally, for $\rho>0$, increasing
$\abs{\mathbf{h}_s^T\mathbf{w}}^2$ strengthens the sensing interference
seen during first-stream decoding and reduces $R_{c,1}(\mathbf{w},\alpha)$.
Finally, increasing $\alpha$ transfers communication power from the first
stream to the second and monotonically decreases $R_{c,1}(\mathbf{w},\alpha)$,
with $\alpha=1$ yielding $R_{c,1}=0$.

\subsection{\gls{sr} After \gls{sic} of the First Stream}
Assuming perfect \gls{sic} of the first communication stream
$\mathbf{s}_{c,1}$~\cite{Zhao2024NearFieldISAC,Mishra2025RSInspiredISAC},
the residual observation at the \gls{bs} is given by
\begin{equation}\label{eq:Y_Sensing_Signal}
\mathbf{Y}_s
=
\sqrt{p_{c,2}}\,\mathbf{h}_c\mathbf{s}_{c,2}^H
+
\sqrt{p_s}\,\beta\,\mathbf{h}_s\mathbf{h}_s^T\mathbf{w}\mathbf{s}_s^H
+
\mathbf{N},
\end{equation}
where $\mathbf{s}_{c,2}$ is treated as interference during
the sensing stage. The \gls{sr} is defined as the sensing \gls{mi} per
symbol, conditioned on the sensing waveform and beamformer
~\cite{Ouyang_MI,TangSensingMI}. Vectorizing
\eqref{eq:Y_Sensing_Signal} gives
\begin{equation}\label{eq:Vec_Y_Sensing_Signal}
\mathrm{vec}(\mathbf{Y}_s)
=
\sqrt{p_s}\,(\mathbf{h}_s^T\mathbf{w})\,
(\mathbf{s}_s^*\otimes\mathbf{h}_s)\beta
+
\mathbf{z}_s,
\end{equation}
where $\mathbf{z}_s\in\mathbb{C}^{NL\times 1}$ denotes the aggregate
interference-plus-noise term with covariance
\begin{equation}
\mathbf{R}_s(\alpha)
=
\mathbf{I}_L\otimes
\bigl(
\mathbf{I}_N+p_{c,2}\mathbf{h}_c\mathbf{h}_c^H
\bigr).
\end{equation}
Accordingly, the achievable \gls{sr} is given by \eqref{eq:Rs_general}.
\begin{figure*}[!t]
\vspace{-0.01em}
\begin{equation}
\label{eq:Rs_general}
R_s(\mathbf{w},\alpha)
=
\frac{1}{L}\log_2\!\left(
1+p_sL\alpha_s\abs{\mathbf{h}_s^T\mathbf{w}}^2
\norm{\mathbf{h}_s}^2
\left[
1-\frac{p_{c,2}\rho\norm{\mathbf{h}_c}^2}
{1+p_{c,2}\norm{\mathbf{h}_c}^2}
\right]
\right).
\end{equation}
\hrulefill
\vspace{-1.2em}
\end{figure*}
\begin{anup}
It follows from \eqref{eq:Rs_general} that the \gls{sr} increases
with the sensing illumination gain
$\abs{\mathbf{h}_s^T\mathbf{w}}^2$ and the channel strength
$\norm{\mathbf{h}_s}^2$, while it is degraded by the undecoded second
communication stream through the overlap factor $\rho$. In particular,
when $\rho$ is small, the corresponding interference penalty is weak.

\par Following the sensing stage, the \gls{bs} forms the \gls{mmse}
estimate $\hat{\beta}$ of the target-response coefficient $\beta$.
Let $\widetilde{\beta}\triangleq\beta-\hat{\beta}$ denote the
estimation error. Its variance, derived in
Appendix~\ref{app:mmse_beta_residual}, is
\begin{equation}
\label{eq:beta_mmse}
\sigma_{\beta|\mathbf{Y}_s}^2
=
\alpha_s2^{-LR_s(\mathbf{w},\alpha)}.
\end{equation}
The resulting residual sensing-echo covariance is characterized in the
following proposition.
\begin{proposition}
\label{prop:Ce_closed}
Under \gls{mmse}-based estimation of the target response, the
symbol-level covariance of the residual sensing echo is
\begin{equation}
\label{eq:Ce_maintext}
\mathbf{C}_e(\mathbf{w},\alpha)
=
p_s\,\alpha_s\,2^{-L R_s(\mathbf{w},\alpha)}
\abs{\mathbf{h}_s^T\mathbf{w}}^2
\mathbf{h}_s\mathbf{h}_s^H.
\end{equation}
\end{proposition}
\begin{proof}
The proof follows from the \gls{mmse} error variance in
\eqref{eq:beta_mmse} and is provided in
Appendix~\ref{app:mmse_beta_residual}.
\end{proof}
Proposition~\ref{prop:Ce_closed} shows that the residual sensing echo
is confined to the spatial direction of $\mathbf h_s$, while its
magnitude is jointly determined by the sensing illumination and the
posterior uncertainty in $\beta$. The factor
$2^{-LR_s(\mathbf w,\alpha)}$ decreases as the sensing observation
becomes more informative, whereas
$\abs{\mathbf h_s^T\mathbf w}^2$ directly scales the echo power.
Consequently, increasing the sensing illumination improves estimation
accuracy but also strengthens the echo to be cancelled, so the residual
covariance increases at a diminishing rate.

\subsection{Second-Stream \gls{cr}}

After the sensing stage, the \gls{bs} subtracts the sensing echo
reconstructed using $\hat{\beta}$ before decoding $\mathbf{s}_{c,2}$. Unlike a decoded communication codeword, whose symbols are exactly
known from the shared codebook, the sensing echo depends on the unknown
coefficient $\beta$ and can only be reconstructed from its estimate,
thereby leaving residual interference.
The resulting achievable rate is characterized as follows.

\begin{proposition}
\label{prop:R2_closed}
Under \gls{mmse}-based sensing-echo cancellation, the achievable
symbol-level rate of the second communication stream is given by
\eqref{eq:R2_closed_form}.
\begin{figure*}[!t]
\vspace{-0.01em}
\begin{equation}
\label{eq:R2_closed_form}
R_{c,2}(\mathbf w,\alpha)
=
\log_2\!\left(
1+p_{c,2}\norm{\mathbf h_c}^2
\right)
-
\log_2\!\left(
1+
\frac{
p_{c,2}\norm{\mathbf h_c}^2
}{
1+p_{c,2}\norm{\mathbf h_c}^2
}
p_s\alpha_s2^{-LR_s(\mathbf w,\alpha)}
\abs{\mathbf h_s^T\mathbf w}^2
\norm{\mathbf h_s}^2\rho
\right).
\end{equation}
\hrulefill
\vspace{-1.2em}
\end{figure*}
\end{proposition}
\begin{proof}
The proof is provided in Appendix~\ref{app:second_stream_rate}.
\end{proof}

Proposition~\ref{prop:R2_closed} expresses the second-stream
\gls{cr} as the rate attainable under perfect sensing-echo cancellation
minus a penalty due to finite target-response uncertainty. Importantly,
the residual covariance in Proposition~\ref{prop:Ce_closed} alone does
not fully characterize the disturbance seen by the second stream,
since the residual sensing echo is statistically coupled with
$\mathbf s_{c,2}$ through the estimation of $\hat{\beta}$. The rate in
\eqref{eq:R2_closed_form} accounts for this dependence, with the
resulting penalty governed by $2^{-LR_s(\mathbf w,\alpha)}$,
 $\abs{\mathbf h_s^T\mathbf w}^2$, and
overlap $\rho$.
\par In particular, when $\rho=0$, the penalty vanishes and
\eqref{eq:R2_closed_form} reduces to
$\log_2(1+p_{c,2}\norm{\mathbf h_c}^2)$, irrespective of the residual
target-response uncertainty. At the other extreme, $\rho=1$ produces
the strongest coupling between the residual sensing echo and the
communication channel. Increasing $\alpha$ raises the second-stream
power $p_{c,2}$, while also increasing its interference during sensing
and thereby degrading the subsequent sensing-echo cancellation
accuracy. Nevertheless, $R_{c,2}(\mathbf w,\alpha)$ increases
monotonically with $\alpha$, while $R_s(\mathbf w,\alpha)$ decreases
for $\rho>0$, exposing the corresponding \gls{sandc}
trade-off.

\par For comparison, perfect post-sensing cancellation, as adopted in
existing uplink \gls{isac} analyses
\cite{Zhao2024NearFieldISAC,OuyangImpactrevealSIC}, yields
\begin{equation}
\label{eq:R2_ideal}
R_{c,2}^{\mathrm{ideal}}(\alpha)
=
\log_2\!\left(
1+p_{c,2}\norm{\mathbf h_c}^2
\right).
\end{equation}
For any finite \gls{sr}, however,
\eqref{eq:beta_mmse} gives a non-zero posterior uncertainty in
$\beta$. Thus, \eqref{eq:R2_ideal} serves as a
perfect-cancellation benchmark, while
\eqref{eq:R2_closed_form} captures the communication penalty induced
by finite sensing accuracy.
\end{anup}

\subsection{Asymptotic Analysis}\label{subsubsec:asymp_ana}
In this subsection, we investigate the asymptotic behaviour of the \gls{rs}-inspired uplink \gls{isac} scheme in the high-\gls{snr} regime and, for the near-field specialization, in the large-array regime.

\subsubsection{High-\gls{snr} Regime}\label{subsubsec:highsnr}
We first examine the high-\gls{snr} behaviour of the
\gls{rs}-inspired uplink \gls{isac} scheme under a common power-scaling law, in which the \gls{sandc} powers are parameterized as
$p_c=\bar{p}_c\,p$ and $p_s=\bar{p}_s\,p$, respectively, with
$p\to\infty$. Here, $\bar{p}_c,\bar{p}_s>0$ are fixed finite scaling
coefficients that determine the relative power levels of the two
functionalities, while the split factor $\alpha\in(0,1)$ and a sensing beamformer
$\mathbf w$ satisfying $\abs{\mathbf h_s^T\mathbf w}^2>0$ are kept fixed. Under this regime, the powers
allocated to the first and second communication streams are
$p_{c,1}=(1-\alpha)\bar{p}_c\,p$ and
$p_{c,2}=\alpha\bar{p}_c\,p$, respectively.
\par To begin with, we characterize the asymptotic behaviour of the first communication stream through the following corollary.
\begin{corollary}
\label{cor:R1_highsnr}
Under the common power-scaling law, the first-stream \gls{cr},
$R_{c,1}(\mathbf{w},\alpha)$, converges to a finite constant as
$p\to\infty$. In particular, for $\rho<1$,
\begin{equation}
\label{eq:R1_highsnr}
R_{c,1}(\mathbf{w},\alpha)
=
\log_2\!\left(\frac{1}{\alpha}\right)+o(1),
\; p\to\infty,
\end{equation}
and therefore has zero high-\gls{snr} slope. For the fully aligned case
$\rho=1$, the rate also saturates, with the limiting expression given by
\eqref{eq:R1_highsnr_rho1}.
\vspace{-0.2em}
\end{corollary}
\begin{proof}
\renewcommand{\qedsymbol}{}
Substituting the common power-scaling law into
\eqref{eq:R1_closed_form} and retaining the dominant terms as
$p\to\infty$ gives the result. For $\rho<1$, the terms proportional to
$(1-\rho)$ dominate, yielding \eqref{eq:R1_highsnr}. On the other hand, for $\rho=1$,
these terms vanish and the resulting dominant-order balance yields
\eqref{eq:R1_highsnr_rho1}.
\vspace{-0.5em}
\end{proof}
\begin{figure*}[!t]
\begin{equation}
\label{eq:R1_highsnr_rho1}
R_{c,1}(\mathbf{w},\alpha)
=
\log_2\!\Bigg(
1+
\frac{(1-\alpha)\norm{\mathbf{h}_c}^2}{
\alpha\norm{\mathbf{h}_c}^2
+\dfrac{\bar{p}_s}{\bar{p}_c}\alpha_s
\abs{\mathbf{h}_s^T\mathbf{w}}^2\norm{\mathbf{h}_s}^2
}
\Bigg)+o(1),
\; \rho=1,\; p\to\infty.
\end{equation}
\hrulefill
\vspace{-0.2em}
\end{figure*}
Corollary~\ref{cor:R1_highsnr} shows that, irrespective of the channel overlap, the first communication stream saturates at high \gls{snr} and therefore contributes no asymptotic rate slope. We next characterize the asymptotic behaviour of the \gls{sr}. The corresponding high-\gls{snr} result is stated in the following corollary.
\begin{corollary}
\label{cor:Rs_highsnr}
Under the common power-scaling law, the \gls{sr}
$R_s(\mathbf{w},\alpha)$ grows logarithmically with $p$ for $\rho<1$
according to \eqref{eq:Rs_highsnr_slope},
\begin{figure*}[!t]
\begin{equation}
\label{eq:Rs_highsnr_slope}
R_s(\mathbf{w},\alpha)
=
\frac{1}{L}\log_2 p
+
\frac{1}{L}\log_2\!\Big(
\bar{p}_sL\alpha_s\abs{\mathbf{h}_s^T\mathbf{w}}^2
\norm{\mathbf{h}_s}^2(1-\rho)
\Big)
+o(1),
\; \rho<1,\; p\to\infty,
\end{equation}
\vspace{-1.2em}
\end{figure*}
and therefore has high-\gls{snr} slope $1/L$. By contrast, for the
fully aligned case $\rho=1$, the \gls{sr} saturates, with the limiting
expression given by \eqref{eq:Rs_highsnr_rho1}.
\begin{figure*}[!t]
\begin{equation}
\label{eq:Rs_highsnr_rho1}
R_s(\mathbf{w},\alpha)
=
\frac{1}{L}\log_2\!\left(
1+\frac{\bar{p}_sL\alpha_s\abs{\mathbf{h}_s^T\mathbf{w}}^2
\norm{\mathbf{h}_s}^2}
{\alpha\bar{p}_c\norm{\mathbf{h}_c}^2}
\right)
+o(1),
\; \rho=1,\; p\to\infty.
\end{equation}
\hrulefill
\vspace{-1.2em}
\end{figure*}
\end{corollary}

\begin{proof}
\renewcommand{\qedsymbol}{}
Substituting the common power-scaling law into \eqref{eq:Rs_general}
and retaining the dominant terms as $p\to\infty$ gives the result. For $\rho<1$, the bracketed term in \eqref{eq:Rs_general} converges to $1-\rho$, yielding \eqref{eq:Rs_highsnr_slope}. By contrast, for $\rho=1$, the same
term decays as $1/p$, cancelling the outer power-scaling factor and yielding the finite limit in \eqref{eq:Rs_highsnr_rho1}.
\end{proof}
\par Corollary~\ref{cor:Rs_highsnr} therefore shows that the \gls{sr} has high-\gls{snr} slope $1/L$ whenever the \gls{sandc} channels are not fully aligned, whereas it saturates for $\rho=1$. \AM{Next, we characterize the asymptotic behaviour of the second communication stream.}
\begin{anup}
\begin{corollary}
\label{cor:R2_highsnr}
Under the common power-scaling law, the second-stream \gls{cr},
$R_{c,2}(\mathbf{w},\alpha)$, grows logarithmically with $p$ for
$\rho<1$ according to \eqref{eq:R2_highsnr_slope}, and therefore has
unit high-\gls{snr} slope. By contrast, for $\rho=1$, the second-stream
rate saturates, with the limiting expression given by
\eqref{eq:R2_highsnr_rho1}.
\end{corollary}
\begin{proof}
\renewcommand{\qedsymbol}{}
The result follows by substituting the common power-scaling law into
\eqref{eq:R2_closed_form} and using the high-\gls{snr}
behaviour of $R_s(\mathbf w,\alpha)$ from
Corollary~\ref{cor:Rs_highsnr}. For $\rho<1$, the residual-dependent
term converges to a finite constant, yielding unit slope, whereas for
$\rho=1$ it scales with $p$ and cancels the logarithmic growth of the
first term, yielding the finite limit in
\eqref{eq:R2_highsnr_rho1}.
\end{proof}

\begin{figure*}[!t]
\vspace{0.2em}
\begin{equation}
\label{eq:R2_highsnr_slope}
R_{c,2}(\mathbf{w},\alpha)
=
\log_2 p
+
\log_2\!\left(
\alpha\bar{p}_c\norm{\mathbf{h}_c}^2
\frac{L(1-\rho)}
{L(1-\rho)+\rho}
\right)
+o(1),
\;\rho<1,\;p\to\infty.
\end{equation}
\vspace{-1.2em}
\end{figure*}
\begin{figure*}[!t]
\begin{equation}
\label{eq:R2_highsnr_rho1}
R_{c,2}(\mathbf{w},\alpha)
=
\log_2\!\left(
L+
\frac{\alpha\bar{p}_c\norm{\mathbf{h}_c}^2}
{\bar{p}_s\alpha_s
\abs{\mathbf{h}_s^T\mathbf{w}}^2\norm{\mathbf{h}_s}^2}
\right)
+o(1),
\;\rho=1,\;p\to\infty.
\end{equation}
\hrulefill
\vspace{-1.2em}
\end{figure*}
\par Corollary~\ref{cor:R2_highsnr} shows that the second communication
stream has unit high-\gls{snr} slope whenever the \gls{sandc} channels
are not fully aligned, whereas it saturates for $\rho=1$.
Together with Corollary~\ref{cor:R1_highsnr}, this implies that, for
$\rho<1$, the aggregate \gls{cr}
$R_c(\mathbf w,\alpha)=R_{c,1}(\mathbf w,\alpha)
+R_{c,2}(\mathbf w,\alpha)$ satisfies
\begin{equation}
\label{eq:Rc_highsnr}
R_c(\mathbf w,\alpha)
=
\log_2 p+O(1),
\qquad
\rho<1,\;p\to\infty,
\end{equation}
and therefore has unit high-\gls{snr} slope. Combining this with
Corollary~\ref{cor:Rs_highsnr}, for $\rho<1$ the
\gls{rs}-inspired uplink scheme achieves \gls{sr} and aggregate
\gls{cr} slopes of $1/L$ and $1$, respectively, whereas for $\rho=1$ both rates saturate. These conclusions apply to both near- and
far-field propagation, with the regime affecting the rate
offsets rather than the asymptotic slopes.
\par For comparison, we evaluate the high-\gls{snr} slopes of the two
endpoint uplink \gls{sic} strategies and uplink \gls{fdsac} under the
same common power-scaling law. As summarized in
Table~\ref{tab:uplink_slopes_rs}, for $\rho<1$, both endpoint
\gls{sic} strategies and the \gls{rs}-inspired scheme achieve the same
asymptotic slope pair, namely an \gls{sr} slope of $1/L$ and a
\gls{cr} slope of one. By contrast, uplink \gls{fdsac} achieves slopes
$\kappa/L$ and $1-\kappa$ for \gls{sandc}, respectively, where
$\kappa\in[0,1]$ denotes the fraction of bandwidth allocated to sensing.
The reduction follows from the explicit partition of bandwidth between
the two functionalities.
\end{anup}
\begin{table}[!h]
\vspace{-0.6cm}
\centering
\caption{High-\gls{snr} slopes, $\rho < 1$.}
\label{tab:uplink_slopes_rs}
\renewcommand{\arraystretch}{1.15}
\setlength{\tabcolsep}{5pt}
\begin{tabular}{|c|c|c|}
\hline
Scheme & \gls{sr} Slope & \gls{cr} Slope \\
\hline
S-C & ${1}/{L}$ & $1$ \\
\hline
C-C & ${1}/{L}$ & $1$ \\
\hline
\gls{rs}-inspired & ${1}/{L}$ & $1$ \\
\hline
\gls{fdsac} & $\kappa/L$ & $1-\kappa$ \\
\hline
\end{tabular}\vspace{-1cm}
\end{table}
\begin{remark}
\label{rem:residual_covariance_highsnr}
For $\rho<1$, \gls{mmse}-based sensing-echo cancellation preserves the
leading high-\gls{snr} slopes while changing the rate offsets and,
consequently, the attainable \gls{sandc} trade-off. For
$\rho=1$, the distinction is more fundamental: the second-stream
\gls{cr} saturates, whereas under perfect post-sensing cancellation it
retains unit high-\gls{snr} slope. Thus, perfect sensing cancellation
can alter even the asymptotic communication behaviour when the overlap factor $\rho=1$.
\end{remark}

\subsubsection{Large-Array Regime}\label{subsubsec:largearray}
The large-array behaviour of the \gls{rs}-inspired uplink \gls{isac}
scheme is next examined under the near-field channel specialization by
letting $N_y,N_z\to\infty$. Following the large-array treatment
in~\cite{Zhao2024NearFieldISAC}, the near-field channel norms admit the
geometric representation
\begin{equation}
\label{eq:geom_norm_generic}
\norm{\mathbf{h}_i}^2
=
\frac{\zeta}{4\pi}
\sum_{y\in \mathcal{Y}_i}\sum_{z\in \mathcal{Z}_i}\delta_i(y,z),
\; i\in\{c,s\},
\end{equation}
where $\mathcal{Y}_i$ and $\mathcal{Z}_i$ denote the normalized
aperture-coordinate sets induced by the geometry of node $i$, and
$\delta_i(y,z)$ is the corresponding aperture kernel defined
in~\cite{Zhao2024NearFieldISAC}. As the array aperture grows, the channel
norms converge to
\begin{equation}
\label{eq:norm_limit}
\lim_{N_y,N_z\to\infty}\norm{\mathbf{h}_c}^2
=
\lim_{N_y,N_z\to\infty}\norm{\mathbf{h}_s}^2
=
\frac{\zeta}{3}.
\end{equation}
Moreover, for a normalized sensing-beamformer sequence for which the
effective sensing gain converges, define
\begin{equation}
D(\mathbf w)\triangleq
\lim_{N_y,N_z\to\infty}|\mathbf h_s^T\mathbf w|^2,
\end{equation}
which satisfies $0\le D(\mathbf w)\le\zeta/3$ by Cauchy--Schwarz.
\begin{anup}
The limiting \gls{sandc} channel overlap is denoted by
\begin{equation}
\label{eq:Cccf_def}
C_{\rho}
\triangleq
\lim_{N_y,N_z\to\infty}\rho \in [0,1].
\end{equation}
The existence of this limit follows from \eqref{eq:norm_limit}: since
$\norm{\mathbf h_i}^2\to\zeta/3<\infty$, $i\in\{c,s\}$, the
Cauchy--Schwarz inequality gives
\[
\sum_n |h_{c,n}||h_{s,n}|
\leq
\norm{\mathbf h_c}\norm{\mathbf h_s}<\infty,
\]
and hence $\mathbf h_c^H\mathbf h_s$ converges absolutely. Consequently,
$C_\rho$ is well defined, with its value determined by the relative
CU--target geometry through the limiting spatial overlap of their
channels; its closed-form evaluation is generally intractable.
\par Substituting
\eqref{eq:norm_limit}--\eqref{eq:Cccf_def} into
\eqref{eq:R1_closed_form} and \eqref{eq:Rs_general} yields
\eqref{eq:R1_largearray} and \eqref{eq:Rs_largearray}, respectively.
Hence, both $R_{c,1}(\mathbf{w},\alpha)$ and
$R_s(\mathbf{w},\alpha)$ remain finite as the array size grows.
\begin{figure*}
\begin{equation}
\label{eq:R1_largearray}
R_{c,1}^{(\infty)}(\mathbf{w},\alpha)
\triangleq
\lim_{N_y,N_z\to\infty}R_{c,1}(\mathbf{w},\alpha)
=
\log_2\!\left(
1+
p_{c,1}
\frac{
\frac{\zeta}{3}
\left[
1+
p_s\alpha_s D(\mathbf{w})\frac{\zeta}{3}(1-C_{\rho})
\right]
}{
1+
p_{c,2}\frac{\zeta}{3}
+
p_s\alpha_s D(\mathbf{w})\frac{\zeta}{3}
+
p_{c,2}p_s\alpha_s D(\mathbf{w})\frac{\zeta^2}{9}(1-C_{\rho})
}
\right).
\end{equation}
\hrulefill
\vspace{-1.2em}
\end{figure*}
\begin{figure*}
\begin{equation}
\label{eq:Rs_largearray}
R_s^{(\infty)}(\mathbf{w},\alpha)
\triangleq
\lim_{N_y,N_z\to\infty}R_s(\mathbf{w},\alpha)
=
\frac{1}{L}\log_2\!\left(
1+
p_sL\alpha_s D(\mathbf{w})\frac{\zeta}{3}
\left[
1-
\frac{
p_{c,2}C_{\rho}\frac{\zeta}{3}
}{
1+p_{c,2}\frac{\zeta}{3}
}
\right]
\right).
\end{equation}
\hrulefill
\vspace{-1.2em}
\end{figure*}
Similarly, substituting
\eqref{eq:norm_limit}--\eqref{eq:Cccf_def} and
\eqref{eq:Rs_largearray} into \eqref{eq:R2_closed_form} yields the
large-array second-stream rate in \eqref{eq:R2_largearray}.
\begin{figure*}[!t]
\begin{equation}
\label{eq:R2_largearray}
\begin{split}
R_{c,2}^{(\infty)}(\mathbf{w},\alpha)
&\triangleq
\lim_{N_y,N_z\to\infty}R_{c,2}(\mathbf{w},\alpha)
=
\log_2\!\left(
1+p_{c,2}\frac{\zeta}{3}
\right)
-
\log_2\!\left(
1+
\frac{
p_{c,2}\frac{\zeta}{3}
}{
1+p_{c,2}\frac{\zeta}{3}
}
\frac{
p_s\alpha_sD(\mathbf{w})\frac{\zeta}{3}C_{\rho}
}{
1+
p_sL\alpha_sD(\mathbf{w})\frac{\zeta}{3}
\left[
1-
\frac{
p_{c,2}C_{\rho}\frac{\zeta}{3}
}{
1+p_{c,2}\frac{\zeta}{3}
}
\right]
}
\right).
\end{split}
\end{equation}
\hrulefill
\vspace{-1.2em}
\end{figure*}
Since $0\le D(\mathbf{w})\le\zeta/3$ and
$0\le C_{\rho}\le1$, \eqref{eq:R2_largearray} remains finite.
Moreover, it is non-increasing in $D(\mathbf{w})$, with the largest
and smallest values attained at $D(\mathbf{w})=0$ and
$D(\mathbf{w})=\zeta/3$, respectively. Consequently, the \gls{sr}
and both communication-stream rates remain finite in the near-field
large-array regime.
\end{anup}
\vspace{-0.5cm}
\begin{anup}
\section{Rate-Region Characterization}
\label{sec:rate_region_characterization}
In this section, we characterize the \gls{sr}--\gls{cr} trade-off of the
\gls{rs}-inspired uplink \gls{isac} scheme under joint optimization of
the sensing beamformer $\mathbf w$ and the \gls{rs} split factor
$\alpha$. The aggregate \gls{cr} is
\begin{equation}
R_c(\mathbf w,\alpha)
=
R_{c,1}(\mathbf w,\alpha)
+
R_{c,2}(\mathbf w,\alpha),
\end{equation}
while the corresponding \gls{sr} is $R_s(\mathbf w,\alpha)$.
\par From \eqref{eq:R1_closed_form}, \eqref{eq:Rs_general}, and
\eqref{eq:R2_closed_form}, it follows that the sensing beamformer
$\mathbf w$ affects all rate expressions only through the scalar term
$\abs{\mathbf h_s^T\mathbf w}^2$. Hence, its effect can be parameterized
by the normalized sensing illumination factor
\begin{equation}
\label{eq:xi_def}
\xi
\triangleq
\frac{\abs{\mathbf h_s^T\mathbf w}^2}
{\norm{\mathbf h_s}^2}
\in[0,1].
\end{equation}
For the considered multi-antenna \gls{bs}, every $\xi\in[0,1]$ is
achievable under the unit-norm beamforming constraint. Hence, the joint beamformer-and-split optimization reduces
exactly to the two scalar variables $(\xi,\alpha)$.
\par Let $R_s(\xi,\alpha)$ and $R_c(\xi,\alpha)$ denote the
corresponding rates obtained by substituting
$\abs{\mathbf h_s^T\mathbf w}^2=\xi\norm{\mathbf h_s}^2$ into the
respective rate expressions. The achievable \gls{rs}-inspired
\gls{sr}--\gls{cr} region is then given by
\begin{equation}
\mathcal{C}_{\mathrm{RS}}
\triangleq
\bigcup_{\substack{0\le\xi\le1,\,0\le\alpha\le1}}
\left\{
(R_s,R_c):
\begin{array}{l}
R_s\le R_s(\xi,\alpha),\\
R_c\le R_c(\xi,\alpha)
\end{array}
\right\}.
\end{equation}
\par Subsequently, the Pareto boundary of $\mathcal{C}_{\mathrm{RS}}$ can be
characterized using the rate-profile method. For a given profile
parameter $\sigma\in[0,1]$, the corresponding boundary point is obtained
from
\begin{equation}
\label{eq:rate_profile_problem_joint}
\begin{aligned}
\max_{\xi,\alpha,R}\quad & R \\
\mathrm{s.t.}\quad
& R_s(\xi,\alpha)\ge \sigma R,\\
& R_c(\xi,\alpha)\ge (1-\sigma)R,\\
& 0\le\xi\le1,\;0\le\alpha\le1.
\end{aligned}
\end{equation}
For each $\sigma$, the split factor $\alpha$ is first optimized for a
given $\xi$, after which the resulting rate-profile value is maximized
over $\xi\in[0,1]$. Since both variables are optimized over their complete
feasible intervals, this nested procedure is equivalent to the joint
optimization in \eqref{eq:rate_profile_problem_joint}.
\par Consistent with this nested characterization, consider a given
sensing illumination factor $\xi\in[0,1]$. For fixed $\xi$, varying the
\gls{rs} split factor $\alpha\in[0,1]$ parametrizes the corresponding
\gls{rs}-inspired operating curve. The endpoint values $\alpha=0$ and $\alpha=1$ reduce
the \gls{rs}-inspired architecture to the  \gls{noma}-inspired \gls{stc} and \gls{cc}
operating modes, respectively, whereas $0<\alpha<1$ realizes message
splitting across the two decoding positions within a single frame. Let
\begin{equation}
A_{\xi}
=
\bigl(R_s(\xi,0),R_c(\xi,0)\bigr),
\,
B_{\xi}
=
\bigl(R_s(\xi,1),R_c(\xi,1)\bigr)
\end{equation}
denote the corresponding endpoint operating points. Here, both endpoint rates are evaluated under the same sensing-echo
cancellation model as the \gls{rs}-inspired receiver.
\par For each $\xi$, the \gls{noma}-inspired endpoint-\gls{sic}
benchmark obtains intermediate operating points by time sharing between
$A_{\xi}$ and $B_{\xi}$. Specifically, for
$\vartheta\in[0,1]$,
\begin{equation}
\label{eq:noma_ts_baseline}
\begin{aligned}
R_s^{\mathrm{NOMA}}(\xi,\vartheta)
&=
\vartheta R_s(\xi,0)
+
(1-\vartheta)R_s(\xi,1),\\
R_c^{\mathrm{NOMA}}(\xi,\vartheta)
&=
\vartheta R_c(\xi,0)
+
(1-\vartheta)R_c(\xi,1).
\end{aligned}
\end{equation}
Hence, the \gls{noma}-inspired decoding-order time-sharing region,
accounting for all feasible sensing illumination factors, is
\begin{equation}
\label{eq:noma_joint_region}
\mathcal{C}_{\mathrm{NOMA}}^{\mathrm{TS}}
=
\bigcup_{\substack{\xi\in[0,1]\\\vartheta\in[0,1]}}
\left\{
(R_s,R_c):
\begin{array}{l}
R_s\le R_s^{\mathrm{NOMA}}(\xi,\vartheta),\\
R_c\le R_c^{\mathrm{NOMA}}(\xi,\vartheta)
\end{array}
\right\}.
\end{equation}
\par This construction parallels the Gaussian uplink
\gls{mac} rate-region characterization, where, for a given physical
channel configuration, the two \gls{sic} decoding orders define the
corner points of the dominant face and intermediate operating points are
obtained through time sharing. Here, for each $\xi$, $A_{\xi}$ and
$B_{\xi}$ play the corresponding endpoint roles, while varying
$\alpha\in[0,1]$ traces the single-frame \gls{rs}-inspired boundary
between them. The outer optimization over
$\xi$ then accounts for all feasible sensing illumination factors. The sensing-matched beamformer
in~\cite{Zhao2024NearFieldISAC} is recovered as the special case
$\xi=1$. The following proposition establishes the resulting containment relation
between the \gls{rs}-inspired and \gls{noma}-inspired regions.
\begin{proposition}
\label{prop:rs_dominates_noma_fixed_xi}
For any $\xi\in[0,1]$ and $L\geq1$, the \gls{rs}-inspired boundary
obtained by varying $\alpha\in[0,1]$ is concave in the
$(R_s,R_c)$ plane between $A_{\xi}$ and $B_{\xi}$. Consequently,
\begin{equation}
\operatorname{conv}\!\left\{A_{\xi},B_{\xi}\right\}
\subseteq
\mathcal{C}_{\mathrm{RS}}(\xi),
\end{equation}
where $\mathcal{C}_{\mathrm{RS}}(\xi)$ denotes the achievable
\gls{rs}-inspired region for a given $\xi$. Since the result holds for
all $\xi\in[0,1]$,
\begin{equation}
\label{eq:rs_contains_noma_joint}
\mathcal{C}_{\mathrm{NOMA}}^{\mathrm{TS}}
\subseteq
\mathcal{C}_{\mathrm{RS}}.
\end{equation}
Hence, the jointly optimized \gls{rs}-inspired region contains the
\gls{noma}-inspired decoding-order time-sharing region while operating
through a single-frame message split.
\end{proposition}
\begin{proof}
The proof is detailed in
Appendix~\ref{app:proof_rs_dominates_noma_fixed_xi}.
\end{proof}
Proposition~\ref{prop:rs_dominates_noma_fixed_xi} establishes that the
containment holds for every sensing illumination factor $\xi$, and
therefore also under joint optimization over $\xi$. Unlike the Gaussian
uplink \gls{mac}, where \gls{rs} recovers the affine dominant face
generated by time sharing between \gls{sic} decoding orders, varying
$\alpha$ in uplink \gls{isac} also changes the \gls{sr} through the
interference produced by the communication stream decoded after sensing.
Consequently, the \gls{rs}-inspired boundary is curved and
lies on or above the corresponding endpoint-\gls{sic} time-sharing
chord.
\end{anup}
\begin{anup}
\section{Numerical Results}\label{sec:num_results}
\begin{table}[!t]
\vspace{-0.5cm}
\centering
\caption{Simulation parameters.}
\label{tab:sim_params}
\footnotesize
\renewcommand{\arraystretch}{1.15}
\setlength{\tabcolsep}{3pt}
\begin{tabular}{l c}
\toprule
Parameter & Value \\
\midrule
Carrier wavelength, \(\lambda\) & \(0.125~\mathrm{m}\) \\
Antenna spacing, \(d\) & \(\lambda/2\) \\
Element area, \(A\) & \(\lambda^2/(4\pi)\) \\
Array size & \(N_y=N_z=15\) \\
Frame length, \(L\) & \(4\) \\
Target reflection power, \(\alpha_s\) & \(1\) \\
Communication power, \(p_c\) & \(60~\mathrm{dB}\) \\
Sensing power, \(p_s\) & \(85~\mathrm{dB}\) \\
CU location, \((r_c,\theta_c,\phi_c)\) & \((10~\mathrm{m},\pi/4,\pi/6)\) \\
Weak-alignment target location & \((5~\mathrm{m},\pi/4,-\pi/6)\) \\
Strong-alignment target location & \((10~\mathrm{m},\pi/4,\pi/6+\Delta_\phi)\) \\
Angular offset, \(\Delta_\phi\) & \(2^\circ\) \\
\bottomrule
\end{tabular}\vspace{-0.3cm}
\end{table}
In this section, we present numerical results to evaluate the
\gls{sandc} trade-off of the \gls{rs}-inspired uplink
\gls{isac} scheme. Unless otherwise stated, we follow the baseline
near-field uplink setup in~\cite{Zhao2024NearFieldISAC}, with the
numerical parameters summarized in Table~\ref{tab:sim_params}. Since the receiver noise is normalized to unit variance, $p_c$ and $p_s$
denote noise-normalized transmit-power parameters. We  next fix
the \gls{cu} at $(r_c,\theta_c,\phi_c)=(10~\mathrm{m},\pi/4,\pi/6)$. To examine the role of \gls{sandc} channel alignment, we consider two target
placements: a weak-alignment placement corresponding to the baseline
geometry, which yields a small channel-correlation factor $\rho$, and a
strong-alignment placement obtained by placing the target close to the
\gls{cu} in range and angle, which yields a larger $\rho$.
For the near- and far-field channel-model comparison, we retain the same
physical geometry and system parameters and vary only the propagation
model as per equation \eqref{eq:nf_ff_element_channel}~\cite{Zhao2024NearFieldISAC}. The resulting far-field \gls{upw} curves
therefore provide a controlled propagation-model baseline for isolating
the effect of spherical-wave near-field propagation.
\par We organize the numerical analysis into two parts. First, we examine
the \gls{sr}--\gls{cr} rate regions. We start with sensing-matched
illumination, corresponding to $\xi=1$, to illustrate the effect of
single-frame message splitting and the differences between near- and
far-field propagation. We then study the rate regions obtained by
jointly optimizing $\xi$ and $\alpha$. Second, we validate the
asymptotic results in Section~\ref{subsec:num_asymptotic_power},
including the high-\gls{snr} behaviour under common power scaling and
the near-field large-array limits in
Section~\ref{subsubsec:asymp_ana}.
\end{anup}
\vspace{-0.1cm}
\subsection{Rate-Region Analyses}\label{subsubsec:rr_sensing_matched}
\AM{\subsubsection{Sensing-Matched Rate-Region} Figs.~\ref{fig:rr_xi1_nf_smallrho}--\ref{fig:rr_xi1_ff_largerho}
show the sensing-matched rate regions under near-field and far-field
propagation for weak ($\rho\ll1$) and strong ($\rho\lesssim1$)
\gls{sandc} channel alignment, respectively. In all cases,
point $\mathrm{A}$ denotes the common \gls{stc} endpoint at
$\alpha=0$.} Since no second communication stream is present at this point,
$\mathrm{A}$ is identical for the \gls{rs}-inspired and
\gls{noma}-inspired schemes under both \gls{mmse}-based and
ideal-cancellation assumptions. Under ideal cancellation, increasing
the role of the \gls{cc} endpoint leads both schemes to the same
$\alpha=1$ point, denoted by $\mathrm{C}$. This point already achieves
the communication-oriented ideal-cancellation limit, so both the
\gls{rs}-inspired and \gls{noma}-inspired regions share the same
$R_s=0$ intercept, denoted by $\mathrm{D}$. Point $\mathrm{D}$ also
coincides with the \gls{fdsac} point at $\kappa=0$, since both
correspond to the pure communication-only rate.
\par Under residual-aware sensing-echo cancellation, where the target
response is estimated using the \gls{mmse} estimator and the resulting
estimation residual is explicitly accounted for, point $\mathrm{B}$
denotes the \gls{cc} endpoint corresponding to $\alpha=1$. It lies
below the ideal-cancellation endpoint $\mathrm{C}$ because the
remaining target-response uncertainty penalizes the communication
stream decoded after sensing. For the \gls{noma}-inspired benchmark,
$\mathrm{B}$ is the communication-oriented endpoint, since only the
two endpoint modes $\alpha=0$ and $\alpha=1$ are available;
consequently, its $R_s=0$ intercept is $\mathrm{F}$ through downward
closure. The \gls{rs}-inspired scheme also contains $\mathrm{B}$ as a
feasible endpoint, but can additionally realize intermediate operating
points through an interior split $\alpha\in(0,1)$. For weak channel alignment, the effect of the split on the aggregate \gls{cr} is very small, and the maximum communication rate is attained at the $\alpha=1$ endpoint, so that $\mathrm{E}$ and $\mathrm{F}$ coincide. Accordingly, the \gls{rs}-inspired boundary is almost indistinguishable from the \gls{noma}-inspired time-sharing chord. By contrast, under strong channel alignment, an interior split can
increase the aggregate \gls{cr} above the $\alpha=1$ endpoint. The
\gls{rs}-inspired region therefore reaches the $R_s=0$ axis at the
higher point $\mathrm{E}$, while the \gls{noma}-inspired region
terminates at $\mathrm{F}$. The resulting separation is visible under
both near-field and far-field propagation regimes.
\begin{figure}[!t]
    \centering
    \includegraphics[width=\linewidth]{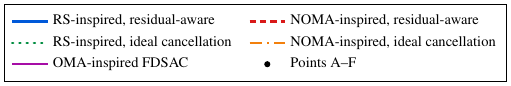}\\[0.5pt]

    \subfloat[Near-field, $\rho\ll1$.\label{fig:rr_xi1_nf_smallrho}]{
        \includegraphics[width=0.47\linewidth]{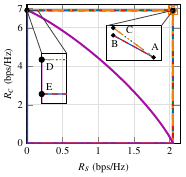}
    }
    \hfill
    \subfloat[Near-field, $\rho\lesssim1$.\label{fig:rr_xi1_nf_largerho}]{
        \includegraphics[width=0.47\linewidth]{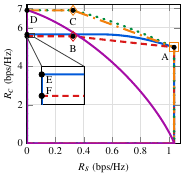}
    }\\[2pt]

    \subfloat[Far-field, $\rho\ll1$.\label{fig:rr_xi1_ff_smallrho}]{
        \includegraphics[width=0.47\linewidth]{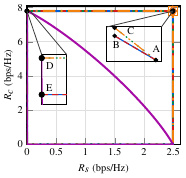}
    }
    \hfill
    \subfloat[Far-field, $\rho\lesssim1$.\label{fig:rr_xi1_ff_largerho}]{
        \includegraphics[width=0.47\linewidth]{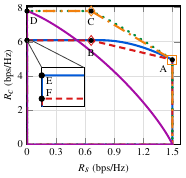}
    }

    \caption{Sensing-matched \gls{sr}--\gls{cr} rate regions under near- and far-field propagation for weak and strong \gls{sandc} channel alignment.}
    \label{fig:rr_xi1}\vspace{-0.5cm}
\end{figure}
\begin{anup}
\par Comparing the near-field and far-field panels shows that the
qualitative rate-region structure is preserved across both propagation
regimes. In particular, weak alignment leads to nearly coincident
\gls{rs}-inspired and \gls{noma}-inspired boundaries, whereas strong
alignment produces a clear gain from single-frame message splitting and
a larger separation between the residual-aware and ideal-cancellation
rate regions. The propagation model primarily changes the channel gains
and, consequently, the resulting rate offsets and the impact of the
remaining target-response uncertainty on post-sensing communication,
while preserving the ordering of the schemes.
\par These results establish three key observations for the
sensing-matched setting under both near-field and far-field propagation.
First, accounting for the remaining target-response uncertainty and its
impact on post-sensing communication reveals a systematic reduction of
the achievable rate region relative to the ideal-cancellation benchmark,
with the effect becoming increasingly pronounced as the
\gls{sandc} channel alignment strengthens. Second, consistent
with Proposition~\ref{prop:rs_dominates_noma_fixed_xi}, the
\gls{rs}-inspired scheme realizes the intermediate
\gls{sandc} trade-off within a single frame through message
splitting, without requiring the inter-frame time sharing used by the
\gls{noma}-inspired benchmark. As a result, its achievable rate region
contains the \gls{noma}-inspired endpoint-\gls{sic} time-sharing region,
with stronger channel alignment making the separation between the two
boundaries increasingly evident. Third, existing \gls{noma}-inspired
uplink \gls{isac} analyses with fixed sensing beamforming and perfect
sensing-echo cancellation~\cite{Zhao2024NearFieldISAC,OuyangImpactrevealSIC}
show that the non-orthogonal rate region contains the \gls{fdsac}
region. The residual-aware analysis demonstrates that this containment
is not generally preserved for the sensing-matched $\xi=1$ region once
the remaining target-response uncertainty is explicitly accounted for.
This, in turn, motivates the jointly optimized rate-region analysis over
$\xi$ and $\alpha$ considered next.
\subsubsection{Jointly Optimized Rate Regions}
Figures~\ref{fig:rr_xiopt_smallrho} and~\ref{fig:rr_xiopt_largerho}
show the achievable near-field rate regions when the sensing illumination
$\xi$ and message split $\alpha$ are jointly optimized for weak
($\rho\ll1$) and strong ($\rho\lesssim1$) \gls{sandc}
channel alignment, respectively. Relative to the sensing-matched cases,
optimizing $\xi$ enlarges the communication-oriented part of the region
by reducing the impact of the remaining target-response uncertainty on
post-sensing communication, while preserving the sensing-oriented
endpoint $\mathrm{A}$. The communication-axis points of the
sensing-matched regions consequently merge into the common
communication-only point $\mathrm{D}$ at $\xi=0$.
\par \par Joint optimization also reveals how the two control variables are
used across channel conditions. Under residual-aware cancellation, the
optimum is governed predominantly by $\xi$, whereas under ideal
cancellation an interior $\alpha$ remains beneficial, particularly under
strong alignment. In both regimes, the optimized \gls{rs}-inspired
region contains the corresponding \gls{noma}-inspired endpoint-\gls{sic}
time-sharing region within a single frame. Moreover, optimizing $\xi$
enlarges both regions, reducing the \gls{fdsac} advantage and recovering
the common communication-only endpoint $\mathrm{D}$.
\begin{figure}[!b]
\vspace{-0.6cm}
    \centering
    \includegraphics[width=\linewidth]{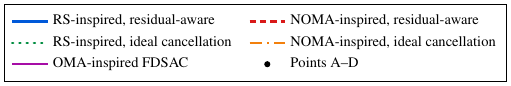}\\[0.5pt]
    \subfloat[$\rho\ll1$.\label{fig:rr_xiopt_smallrho}]{
        \includegraphics[width=0.47\linewidth]{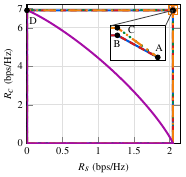}
    }
    \hfill
    \subfloat[$\rho\lesssim1$.\label{fig:rr_xiopt_largerho}]{
        \includegraphics[width=0.47\linewidth]{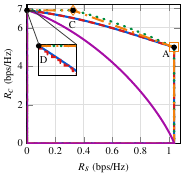}
    }
    \caption{Jointly optimized near-field \gls{sr}--\gls{cr} rate regions
    obtained by optimizing $\xi$ and
    $\alpha$ for weak and strong \gls{sandc} channel alignment.}
    \label{fig:rr_xiopt}
\end{figure}
\end{anup}
\subsection{Asymptotic Analysis}
\label{subsec:num_asymptotic_power}
\subsubsection{High-\gls{snr} Analysis Under Common Power Scaling}
\label{subsubsec:num_highsnr_common_scaling}
We first validate the high-\gls{snr} behaviour derived in Section~\ref{subsubsec:highsnr} under the common power-scaling law $p_c=\bar p_c p,\; p_s=\bar p_s p,\; p\rightarrow\infty$, where \(p\) is the common power-scaling parameter. Unless otherwise stated, we set \(\bar p_c=1\) and \(\bar p_s=10^{25/10}\), so that the sensing power is \(25\) dB larger than the communication power, consistent with the rate-region setup in Table~\ref{tab:sim_params}. The intermediate \gls{rs} split is fixed as \(\alpha=0.5\), and the \gls{fdsac} bandwidth fraction is fixed as \(\kappa=0.5\). This fixed-\(\alpha\) setting is used to directly validate the asymptotic expressions, rather than tracing an optimized rate-profile point.
\par Fig.~\ref{fig:asymptotic_power} reports the \gls{cr} and \gls{sr} as functions of \(p\). For \(\rho<1\), the \gls{cr} curves approach the predicted slope-one behaviour. This agrees with \eqref{eq:R2_highsnr_slope} and \eqref{eq:Rc_highsnr}, which show that the second stream, and hence the aggregate \gls{cr}, scales as \(\log_2 p+O(1)\). The first stream contributes only a finite-rate offset, as shown in \eqref{eq:R1_highsnr}. Consequently, the \(\alpha=0\) endpoint, i.e., S-C, the \(\alpha=1\) endpoint, i.e., C-C, and the intermediate \(\alpha=0.5\) \gls{rs}-inspired scheme have the same aggregate communication high-\gls{snr} slope in this regime. The separation between the ideal-cancellation and residual-aware curves is therefore an asymptotic rate-offset loss. It results from the remaining target-response uncertainty and its coupling to the communication stream decoded after sensing, whose effect is governed by the \gls{sandc} channel overlap $\rho$. For every fixed $\rho<1$, this effect does not alter the leading-order communication
slope. By contrast, \gls{fdsac} grows with communication slope $1-\kappa$, which equals $0.5$ for $\kappa=0.5$. For the \gls{sr}, the curves for $\rho<1$ follow the slope predicted by \eqref{eq:Rs_highsnr_slope}. The non-orthogonal schemes
achieve sensing high-\gls{snr} slope $1/L$, whereas \gls{fdsac} achieves the smaller slope $\kappa/L$. The $\alpha=0$ endpoint lies above the $\alpha=1$ endpoint and the intermediate $\alpha=0.5$ \gls{rs}-inspired sensing curves because no undecoded second communication stream remains during sensing when $\alpha=0$. For $\alpha>0$, the sensing stage treats $\mathbf{s}_{c,2}$ as interference, giving rise to the penalty term in \eqref{eq:Rs_general}. Nevertheless, for every fixed $\rho<1$, this penalty affects only the asymptotic offset, while the sensing slope remains $1/L$.
\begin{figure}[!t]
    \centering
    \includegraphics[width=\linewidth]{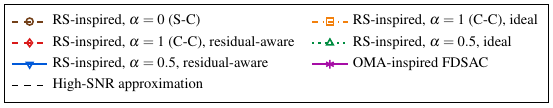}\\[0.2pt]

    \subfloat[CR, \(\rho<1\).\label{fig:cr_asymptotic_power_weak}]{
        \includegraphics[width=0.47\linewidth]{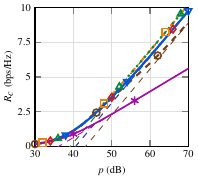}
    }
    \hfill
    \subfloat[SR, \(\rho<1\).\label{fig:sr_asymptotic_power_weak}]{
        \includegraphics[width=0.47\linewidth]{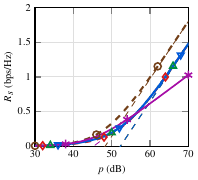}
    }
    \\[0.5em]
    \subfloat[CR, \(\rho=1\).\label{fig:cr_asymptotic_power_strong}]{
        \includegraphics[width=0.47\linewidth]{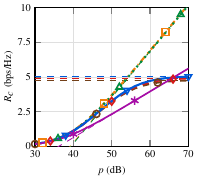}
    }
    \hfill
    \subfloat[SR, \(\rho=1\).\label{fig:sr_asymptotic_power_strong}]{
        \includegraphics[width=0.47\linewidth]{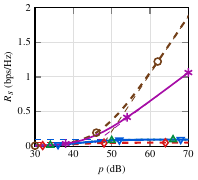}
    }
    \caption{High-\gls{snr} behaviour with \(p_c=\bar p_c p\) and
    \(p_s=\bar p_s p\), with \(\bar p_c=1\),
    \(\bar p_s=10^{25/10}\), \(\alpha=0.5\), and \(\kappa=0.5\).\vspace{-0.3cm}}
    \label{fig:asymptotic_power}
\end{figure}
\par The fully aligned case, $\rho=1$, exhibits a fundamentally
different asymptotic behaviour. In the \gls{cr} panel, the
residual-aware $\alpha=1$ C-C endpoint and the residual-aware
intermediate $\alpha=0.5$ \gls{rs}-inspired curve saturate, consistent
with \eqref{eq:R2_highsnr_rho1}. Under full alignment, the remaining
target-response uncertainty couples directly to the post-sensing
communication direction and prevents the second-stream rate from
growing with $p$. The $\alpha=0$ S-C endpoint also saturates under the
common scaling law because the sensing echo and desired communication
signal scale proportionally with $p$. In contrast, the
ideal-cancellation $\alpha=1$ C-C endpoint and the ideal-cancellation
intermediate $\alpha=0.5$ \gls{rs}-inspired curve retain slope-one
growth because the remaining sensing-echo uncertainty is removed by
assumption.  The corresponding \gls{sr} behaviour for $\rho=1$ follows
\eqref{eq:Rs_highsnr_rho1}. The $\alpha=1$ C-C endpoint and the
fixed-$\alpha$ \gls{rs}-inspired \gls{sr}s saturate because the
undecoded second communication stream is fully aligned with the
sensing channel and scales with $p$. By contrast, the $\alpha=0$ S-C
\gls{sr} continues to grow with slope $1/L$, since no second
communication stream remains during sensing. The \gls{fdsac} sensing
rate grows with slope $\kappa/L$, as expected from bandwidth
partitioning. 
\par Overall, Fig.~\ref{fig:asymptotic_power} confirms a sharp
distinction between the generic $\rho<1$ regime and the singular
fully aligned case. For $\rho<1$, the remaining target-response
uncertainty changes the high-\gls{snr} communication offset but not
the leading \gls{cr} slope, while the interference from the undecoded
second stream similarly changes only the \gls{sr} offset. At
$\rho=1$, these interference terms become asymptotically
slope-limiting for the corresponding \gls{sandc}
operations, producing saturation of the affected rates.

\subsubsection{Large-Array Regime}
\label{subsubsec:largearray_num}
We next validate the large-array behaviour derived in
Section~\ref{subsubsec:largearray}. We consider square \glspl{upa} with
$N_y=N_z$ and plot the rates as functions of the total number of
antennas $N=N_yN_z$. The \gls{rs} split factor and the \gls{fdsac}
bandwidth fraction are fixed as $\alpha=0.5$ and $\kappa=0.5$,
respectively. We consider a representative non-fully-aligned geometry.
The limiting overlap constant $C_\rho$ is evaluated numerically using
the largest simulated array. In the considered setup, the largest
array has $N\approx10^7$, yielding $C_\rho\approx0.565$, which is used
in evaluating the horizontal large-array limits. The aperture
occupation ratio is $\zeta=A/d^2$, and the channel gains satisfy
$\|\mathbf h_c\|^2,\|\mathbf h_s\|^2\rightarrow\zeta/3$, as in
\eqref{eq:norm_limit}. For sensing-matched beamforming,
$\abs{\mathbf h_s^T\mathbf w_s}^2=\|\mathbf h_s\|^2$, and hence this effective sensing gain also converges to $\zeta/3$.
\par Fig.~\ref{fig:large_array} shows the resulting \gls{cr} and \gls{sr}. The exact finite-array curves approach their corresponding horizontal large-array limits, confirming the bounded-rate behaviour predicted by \eqref{eq:R1_largearray}, \eqref{eq:Rs_largearray}, and \eqref{eq:R2_largearray}. This saturation is a direct consequence of the aperture-aware near-field channel model, under which the channel gains converge to finite limits as the array aperture grows. Accordingly, all considered schemes approach finite \gls{sandc} rates as $N\rightarrow\infty$. For the \gls{cr}, the residual-aware curves converge below their ideal-cancellation counterparts, showing that increasing the array size does not eliminate the communication penalty associated with the remaining target-response uncertainty. For the \gls{sr}, the $\alpha=0$ S-C endpoint attains the largest limit because no undecoded second communication stream remains during sensing, whereas the $\alpha=1$ and $\alpha=0.5$ cases retain the corresponding sensing-interference penalty. These results confirm that the \gls{rs}-inspired formulation, its \gls{stc} and \gls{cc} endpoint modes, and the \gls{fdsac} benchmark all inherit the finite large-array behaviour imposed by the physically consistent near-field channel model.
\begin{figure}[!t]
    \centering
    \includegraphics[width=\linewidth]{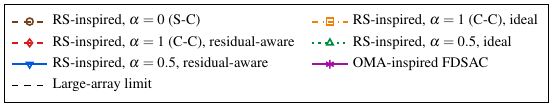}\\[0.5pt]
    \subfloat[\gls{cr}.\label{fig:large_array_cr}]{
        \includegraphics[width=0.474\linewidth]{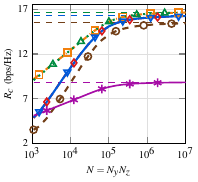}
    }
    \hfill
    \subfloat[\gls{sr}.\label{fig:large_array_sr}]{
        \includegraphics[width=0.47\linewidth]{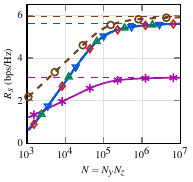}
    }
    \caption{Large-array behaviour for a representative geometry with \(C_\rho\approx0.565\), \(p_c=60\) dB and \(p_s=85\) dB, and \(\kappa=0.5\).\vspace{-2cm}}
    \label{fig:large_array}
\end{figure}
\begin{anup}
\vspace{-.5cm}
\section{Conclusion}\label{sec:conclusion}
This paper characterized the Pareto-optimal \gls{cr}--\gls{sr} rate
region of residual-aware \gls{rs}-inspired uplink \gls{isac} by jointly
optimizing the sensing illumination and communication-message split.
Closed-form \gls{cr} and \gls{sr} expressions were derived while
accounting for the remaining target-response uncertainty and its
statistical coupling to post-sensing communication. The rate-region
analysis proved that the resulting \gls{rs}-inspired region contains the
\gls{noma}-inspired endpoint-\gls{sic} time-sharing region. It further showed that residual sensing-echo interference can break
\gls{oma} containment under fixed illumination, while joint
sensing-illumination optimization enlarges the achievable region and
recovers the maximum \gls{cr} at zero \gls{sr}. The high-\gls{snr} analysis identified fundamentally different asymptotic behaviour under non-fully-aligned and fully aligned \gls{sandc} channels, while the near-field large-array analysis established finite asymptotic rate limits under aperture-aware propagation.
Numerical results validated the rate-region findings under both
near- and far-field propagation and confirmed the high-\gls{snr} and
near-field large-array analyses. Future work may consider fully
non-orthogonal multi-user and multi-target extensions, imperfect channel
knowledge, and other practical extensions.
\end{anup}
\appendices
\section{Proof of Proposition \ref{prop:R1_closed}}
\label{app:first_stream_rate}
This appendix derives a closed-form expression for the scalar term $\mathbf{h}_c^H\mathbf{R}_1^{-1}(\mathbf{w},\alpha)\mathbf{h}_c$. Starting from \eqref{eq:R1_cov}, 
define $a \triangleq p_{c,2},\;
b \triangleq p_s\alpha_s\abs{\mathbf{h}_s^T\mathbf{w}}^2,$ and let $A \triangleq \norm{\mathbf{h}_c}^2,\,
B \triangleq \norm{\mathbf{h}_s}^2,\,
\chi \triangleq \mathbf{h}_c^H\mathbf{h}_s.$ Then $\mathbf{R}_1(\mathbf{w},\alpha)$ can be written as
\begin{equation}
\mathbf{R}_1
=
\mathbf{I}_N+\mathbf{U}\mathbf{U}^H,
\end{equation}
where $\mathbf{U}
\triangleq
\begin{bmatrix}
\sqrt{a}\,\mathbf{h}_c & \sqrt{b}\,\mathbf{h}_s
\end{bmatrix}
\in\mathbb{C}^{N\times 2}.$
Applying the Woodbury identity yields
\begin{equation}
\label{eq:app_R1_inv}
\mathbf{R}_1^{-1}
=
\mathbf{I}_N
-
\mathbf{U}\bigl(\mathbf{I}_2+\mathbf{U}^H\mathbf{U}\bigr)^{-1}\mathbf{U}^H,
\end{equation}
where
\begin{equation}
\label{eq:app_2x2_inv}
\bigl(\mathbf{I}_2+\mathbf{U}^H\mathbf{U}\bigr)^{-1}
=
\frac{1}{\Delta}
\begin{bmatrix}
1+bB & -\sqrt{ab}\,\chi\\
-\sqrt{ab}\,\chi^* & 1+aA
\end{bmatrix},
\end{equation}
with $\Delta
\triangleq
(1+aA)(1+bB)-ab\abs{\chi}^2.$ Substituting \eqref{eq:app_2x2_inv} into \eqref{eq:app_R1_inv} and evaluating the resulting quadratic form gives
\begin{equation}
\label{eq:app_q_closedform}
\mathbf{h}_c^H\mathbf{R}_1^{-1}(\mathbf{w},\alpha)\mathbf{h}_c
=
\frac{
A(1+bB)-b\abs{\chi}^2
}{
(1+aA)(1+bB)-ab\abs{\chi}^2
}.
\end{equation}
Using the channel-correlation factor
\begin{equation}
\rho
=
\frac{\abs{\mathbf{h}_c^H\mathbf{h}_s}^2}{\norm{\mathbf{h}_c}^2\norm{\mathbf{h}_s}^2}
=
\frac{\abs{\chi}^2}{AB},
\end{equation}
\eqref{eq:app_q_closedform} can be equivalently expressed as
\begin{equation}
\label{eq:app_q_rho}
\mathbf{h}_c^H\mathbf{R}_1^{-1}(\mathbf{w},\alpha)\mathbf{h}_c
=
\frac{
A\bigl[1+bB(1-\rho)\bigr]
}{
1+aA+bB+abAB(1-\rho)
}.
\end{equation}
\vspace{-0.6cm}
\begin{anup}
\section{Proof of Proposition \ref{prop:Ce_closed}}
\label{app:mmse_beta_residual}

From \eqref{eq:Vec_Y_Sensing_Signal}, the sensing observation can be
written as
\begin{equation}
\mathbf y_s
\triangleq
\mathrm{vec}(\mathbf Y_s)
=
\mathbf g(\mathbf w)\beta+\mathbf z_s,
\end{equation}
where $\mathbf g(\mathbf w)
\triangleq
\sqrt{p_s}\,(\mathbf h_s^T\mathbf w)
(\mathbf s_s^*\otimes\mathbf h_s)$. Since $\beta\sim\mathcal{CN}(0,\alpha_s)$ and the observation model is
jointly Gaussian, the posterior mean is the optimal \gls{mmse}
estimator of $\beta$, given by
\begin{equation}
\label{eq:app_beta_hat}
\hat{\beta}
=
\alpha_s\mathbf g^H(\mathbf w)
\left(
\alpha_s\mathbf g(\mathbf w)\mathbf g^H(\mathbf w)
+
\mathbf R_s(\alpha)
\right)^{-1}
\mathbf y_s.
\end{equation}
Applying the matrix inversion lemma gives
\begin{equation}
\label{eq:app_beta_hat_mmse}
\hat{\beta}
=
\frac{
\alpha_s\,\mathbf g^H(\mathbf w)\mathbf R_s^{-1}(\alpha)\mathbf y_s
}{
1+\alpha_s\,\mathbf g^H(\mathbf w)\mathbf R_s^{-1}(\alpha)\mathbf g(\mathbf w)
}.
\end{equation}
Hence, with
$\tilde{\beta}\triangleq\beta-\hat{\beta}$, the
\gls{mmse} error variance is
\begin{equation}
\label{eq:app_beta_var}
\sigma_{\beta\mid\mathbf Y_s}^2(\mathbf w,\alpha)
=
\frac{
\alpha_s
}{
1+
\alpha_s\mathbf g^H(\mathbf w)
\mathbf R_s^{-1}(\alpha)\mathbf g(\mathbf w)
}.
\end{equation}
Using the \gls{sr} expression in \eqref{eq:Rs_general}, we obtain
\begin{equation}
\label{eq:app_beta_var_sr}
\sigma_{\beta\mid\mathbf Y_s}^2(\mathbf w,\alpha)
=
\alpha_s2^{-L R_s(\mathbf w,\alpha)}.
\end{equation}

After reconstructing the sensing echo using $\hat{\beta}$, the residual
sensing signal over the frame is
\begin{equation}
\mathbf E
=
\sqrt{p_s}\,\tilde{\beta}\,
\mathbf h_s\mathbf h_s^T\mathbf w\mathbf s_s^H.
\end{equation}
The residual at symbol $\ell$ is therefore
\begin{equation}
\mathbf e_\ell
=
\sqrt{p_s}\,\tilde{\beta}\,
(\mathbf h_s^T\mathbf w)s_{s,\ell}^*\mathbf h_s,
\qquad
\ell=1,\ldots,L.
\end{equation}
Accordingly,
\begin{align}
\mathbf C_{e,\ell}(\mathbf w,\alpha)
&\triangleq
\mathbb E\!\left[
\mathbf e_\ell\mathbf e_\ell^H
\right]
\nonumber\\
&=
p_s\,
\sigma_{\beta\mid\mathbf Y_s}^2(\mathbf w,\alpha)
\abs{\mathbf h_s^T\mathbf w}^2
\abs{s_{s,\ell}}^2
\mathbf h_s\mathbf h_s^H.
\end{align}
Since
$\abs{s_{s,\ell}}^2=1$ $\forall \ell$, the residual covariance is
identical across the frame and, using \eqref{eq:app_beta_var_sr}, is
given by
\begin{equation}
\mathbf C_e(\mathbf w,\alpha)
=
p_s\alpha_s2^{-L R_s(\mathbf w,\alpha)}
\abs{\mathbf h_s^T\mathbf w}^2
\mathbf h_s\mathbf h_s^H,
\end{equation}
which proves Proposition~\ref{prop:Ce_closed}. The covariance $\mathbf C_e(\mathbf w,\alpha)$ characterizes the marginal power and spatial structure of the residual sensing echo. Since $\hat{\beta}$ is obtained from an observation that also contains $\mathbf s_{c,2}$, the residual sensing term is generally correlated with the second communication stream. This statistical dependence is accounted for explicitly in the derivation of the second-stream
\gls{cr} in Appendix~\ref{app:second_stream_rate}.
\vspace{-0.2cm}
\section{Proof of Proposition~\ref{prop:R2_closed}}
\label{app:second_stream_rate}

Let $x_{2,\ell}\triangleq s_{c,2,\ell}^*$ and define
\begin{equation}
\mathbf u
\triangleq
\sqrt{p_s}\,(\mathbf h_s^T\mathbf w)\mathbf h_s,
\quad
\mathbf R_0
\triangleq
\mathbf I_N+p_{c,2}\mathbf h_c\mathbf h_c^H.
\end{equation}
The $\ell$th symbol-level observation after sensing-echo cancellation is
\begin{equation}
\label{eq:app_y2_symbol}
\mathbf y_{2,\ell}
=
\sqrt{p_{c,2}}\mathbf h_c x_{2,\ell}
+
\mathbf u s_{s,\ell}^*\widetilde{\beta}
+
\mathbf n_\ell,
\end{equation}
where $\widetilde{\beta}=\beta-\hat{\beta}$. From
Proposition~\ref{prop:Ce_closed}, we have
\begin{equation}
\label{eq:app_Ce_recall}
\mathbf C_e
=
\sigma_{\beta\mid\mathbf Y_s}^2
\mathbf u\mathbf u^H,
\quad
\sigma_{\beta\mid\mathbf Y_s}^2
=
\alpha_s2^{-LR_s(\mathbf w,\alpha)}.
\end{equation}
\par Since $\hat{\beta}$ is obtained from the frame-level observation
containing the second communication stream, $\widetilde{\beta}$ is
generally correlated with $x_{2,\ell}$. From the \gls{mmse} estimator
in \eqref{eq:app_beta_hat_mmse},
\begin{equation}
\label{eq:app_beta_x_corr}
\mathbb E[
\widetilde{\beta}x_{2,\ell}^*
]
=
-\sigma_{\beta\mid\mathbf Y_s}^2
\sqrt{p_{c,2}}\,
s_{s,\ell}\,
\mathbf u^H\mathbf R_0^{-1}\mathbf h_c.
\end{equation}
The
effective channel of $x_{2,\ell}$ is therefore
\begin{align}
\mathbf h_{2,\mathrm{eff}}
&\triangleq
\mathbb E[
\mathbf y_{2,\ell}x_{2,\ell}^*
]=
\sqrt{p_{c,2}}
\left(
\mathbf I_N-\mathbf C_e\mathbf R_0^{-1}
\right)\mathbf h_c
\nonumber\\
&=
\sqrt{p_{c,2}}
(\mathbf R_0-\mathbf C_e)
\mathbf R_0^{-1}\mathbf h_c.
\label{eq:app_h2_eff}
\end{align}

Moreover, using the orthogonality property of the \gls{mmse}
estimator and \eqref{eq:app_Ce_recall}, the covariance of the
post-cancellation observation is
\begin{equation}
\label{eq:app_Cy2}
\mathbf C_{y_2}
\triangleq
\mathbb E[
\mathbf y_{2,\ell}\mathbf y_{2,\ell}^H
]
=
\mathbf R_0-\mathbf C_e.
\end{equation}
Since $(x_{2,\ell},\mathbf y_{2,\ell})$ are jointly  complex Gaussian, we can write
\begin{equation}
\mathbf y_{2,\ell}
=
\mathbf h_{2,\mathrm{eff}}x_{2,\ell}
+
\mathbf v_\ell,
\end{equation}
where $\mathbf v_\ell$ is independent of $x_{2,\ell}$ and has covariance
\begin{equation}
\mathbf C_v
=
\mathbf C_{y_2}
-
\mathbf h_{2,\mathrm{eff}}
\mathbf h_{2,\mathrm{eff}}^H.
\end{equation}
Hence,
\begin{equation}
R_{c,2}
=
\log_2\!\left(
1+
\mathbf h_{2,\mathrm{eff}}^H
\mathbf C_v^{-1}
\mathbf h_{2,\mathrm{eff}}
\right).
\end{equation}
Since
$\mathbf C_{y_2}
=
\mathbf C_v+
\mathbf h_{2,\mathrm{eff}}\mathbf h_{2,\mathrm{eff}}^H$,
by matrix inversion lemma 
\begin{equation}
\label{eq:app_R2_gaussian}
R_{c,2}
=
-\log_2\!\left(
1-
\mathbf h_{2,\mathrm{eff}}^H
\mathbf C_{y_2}^{-1}
\mathbf h_{2,\mathrm{eff}}
\right).
\end{equation}

From \eqref{eq:app_h2_eff} and \eqref{eq:app_Cy2},
\begin{align}
\mathbf h_{2,\mathrm{eff}}^H
\mathbf C_{y_2}^{-1}
\mathbf h_{2,\mathrm{eff}}
&=
p_{c,2}\mathbf h_c^H
\mathbf R_0^{-1}
(\mathbf R_0-\mathbf C_e)
\mathbf R_0^{-1}\mathbf h_c
\nonumber\\
&=
\frac{P_2}{1+P_2}
-
\frac{P_2Q_e\rho}{(1+P_2)^2},
\label{eq:app_R2_quad}
\end{align}
where $P_2\triangleq p_{c,2}\norm{\mathbf h_c}^2,$ and
\begin{equation}
Q_e\triangleq
p_s\alpha_s2^{-LR_s(\mathbf w,\alpha)}
\abs{\mathbf h_s^T\mathbf w}^2
\norm{\mathbf h_s}^2.
\end{equation}
Here, we used
$\mathbf R_0^{-1}\mathbf h_c
=
\mathbf h_c/(1+P_2)$ and
$\abs{\mathbf h_c^H\mathbf h_s}^2
=
\rho\norm{\mathbf h_c}^2\norm{\mathbf h_s}^2$.
Substituting \eqref{eq:app_R2_quad} into
\eqref{eq:app_R2_gaussian}, and using the definition of $Q_e$, yields
\eqref{eq:R2_closed_form}, which proves
Proposition~\ref{prop:R2_closed}.
\vspace{-0.15cm}
\section{Proof of Proposition~\ref{prop:rs_dominates_noma_fixed_xi}}
\label{app:proof_rs_dominates_noma_fixed_xi}

For a given sensing illumination factor $\xi\in[0,1]$, define
\begin{equation}
P\triangleq p_c\norm{\mathbf h_c}^2,
\quad
V_{\xi}\triangleq
p_s\alpha_s\xi\norm{\mathbf h_s}^4,
\end{equation}
and let $x\triangleq\alpha P$. Using
$\abs{\mathbf h_s^T\mathbf w}^2=\xi\norm{\mathbf h_s}^2$, define
\begin{equation}
a
\triangleq
V_{\xi}\left(1-\frac{x\rho}{1+x}\right)
=
V_{\xi}\frac{1+(1-\rho)x}{1+x}.
\end{equation}
For $\xi\rho=0$, the two endpoint operating points coincide and the
containment follows directly. Hence, consider $\xi\rho>0$, for which
$a$ decreases monotonically with $\alpha$ and parameterizes the
boundary between $A_{\xi}$ and $B_{\xi}$. The \gls{sr} is
\begin{equation}
\label{eq:app_Rs_a}
R_s(a)
=
\frac{1}{L}\log_2(1+La).
\end{equation}

\par Under the residual-aware cancellation model,
$Q_e=V_{\xi}/(1+La)$ and
\begin{equation}
\frac{xQ_e\rho}{1+x}
=
\frac{V_{\xi}-a}{1+La}.
\end{equation}
Combining the two stream rates therefore gives
\begin{equation}
\label{eq:app_Rc_res_a}
\begin{split}
R_c^{\mathrm{res}}(a)
=
K_{\xi}
&-\log_2(1+a)
+\log_2(1+La)\\
&-\log_2\!\left(1+V_{\xi}+(L-1)a\right),
\end{split}
\end{equation}
where
$K_{\xi}\triangleq
\log_2\!\left(1+P+V_{\xi}(1+P(1-\rho))\right)$
is independent of $a$. Letting $y=1+La$, direct differentiation of
\eqref{eq:app_Rs_a} and \eqref{eq:app_Rc_res_a} gives \eqref{eq:app_second_deriv_res}
\begin{figure*}[!t]
\begin{equation}
\label{eq:app_second_deriv_res}
\frac{d^2R_c^{\mathrm{res}}}{dR_s^2}
=
-L^2\ln 2\,y
\left[
\frac{L-1}{(y+L-1)^2}
+
\frac{(L-1)(1+LV_{\xi})}
{\big((L-1)y+1+LV_{\xi}\big)^2}
\right]
\leq 0,
\end{equation}
\hrulefill
\vspace{-1.2em}
\end{figure*}
for $L\geq1$, with strict inequality for $L>1$. Hence, for every fixed $\xi$, the 
\gls{rs}-inspired boundary is concave between $A_{\xi}$ and
$B_{\xi}$, strictly so for $L>1$, and therefore lies on or above the
straight chord joining the endpoints. For perfect post-sensing
cancellation, $R_c^{\mathrm{perf}}(a)=K_{\xi}-\log_2(1+a)$ is likewise
concave in $R_s$, and strictly so for $L>1$. Consequently,
$\operatorname{conv}\{A_{\xi},B_{\xi}\}
\subseteq\mathcal{C}_{\mathrm{RS}}(\xi)$.
Since this holds for every $\xi\in[0,1]$, taking the union over $\xi$
yields
$\mathcal{C}_{\mathrm{NOMA}}^{\mathrm{TS}}
\subseteq\mathcal{C}_{\mathrm{RS}}$.   
\end{anup}
\bibliographystyle{IEEEtran}
\bibliography{reference}
\end{document}